\documentclass[lettersize,journal]{IEEEtran}
\usepackage{amsmath,amsfonts,amsthm,amssymb,bbm}
\usepackage{algorithmicx}
\usepackage{algorithm}
\usepackage{array}
\usepackage[caption=false,font=footnotesize,labelfont=rm,textfont=rm]{subfig}
\usepackage{textcomp}
\usepackage{stfloats}
\usepackage{url}
\usepackage{verbatim}
\usepackage{graphicx}
\usepackage{cite}
\usepackage[colorlinks,linkcolor=blue]{hyperref}

\usepackage{color}
\newtheorem{lemma}{\textbf{\textit{Lemma}}}
\newtheorem{theorem}{\textbf{\textit{Theorem}}}
\newtheorem{corollary}{\textbf{\textit{Corollary}}}
\newtheorem{remark}{\textbf{\textit{Remark}}}
\begin{document}
\title{Interleaved Training for Massive MIMO Downlink via Exploring Spatial Correlation}

\author{ Cheng Zhang,~\IEEEmembership{Member,~IEEE,} Chang Liu,~\IEEEmembership{Student Member,~IEEE,}  Yindi Jing,~\IEEEmembership{Senior Member,~IEEE,}\\Minjie Ding,~\IEEEmembership{Student Member,~IEEE,} Yongming Huang,~\IEEEmembership{Senior Member,~IEEE}
%\thanks{This work was supported in part by the National Natural Science Foundation of China under Grant No. 62271140 and 62225107. (Corresponding author: Y.\ Huang)}
 \thanks{This work was supported in part by the National Natural Science Foundation of China under Grant 62271140 and 62225107, the Fundamental Research Funds for the Central Universities 2242022k60002, the Natural Science Foundation on Frontier Leading Technology Basic Research Project of Jiangsu under Grant BK20222001, and the Major Key Project of PCL. (Corresponding authors: Y. Huang, C. Zhang)}
\thanks{C. Zhang, C. Liu, M. Ding and Y. Huang are with the National Mobile Communication Research Laboratory, Southeast University, Nanjing 210096, China, and also with the Purple Mountain Laboratories, Nanjing 211111, China (e-mail: zhangcheng\_seu, 220210844, 220200816, huangym@seu.edu.cn).}% <-this % stops a space
\thanks{Y. Jing is with the Department of Electrical and Computer Engineering, University of Alberta, Edmonton T6G 1H9, Canada (e-mail: yindi@ualberta.ca).}}

\markboth{IEEE Transactions on WIRELESS COMMUNICATIONS}%
{Submitted paper}
\maketitle

\begin{abstract}

Interleaved training has been studied for single-user and multi-user massive MIMO downlink with either fully-digital or hybrid beamforming. However, the impact of channel correlation on its average training overhead is rarely addressed. In this paper, we explore the channel correlation to improve the interleaved training for single-user massive MIMO downlink. For the beam-domain interleaved training, we propose a modified scheme by optimizing the beam training codebook. The basic antenna-domain interleaved training is also improved by dynamically adjusting the training order of the base station (BS) antennas during the training process based on the values of the already trained channels. Exact and simplified approximate expressions of the average training length are derived in closed-form for the basic and modified beam-domain schemes and the basic antenna-domain scheme in correlated channels. For the modified antenna-domain scheme, a deep neural network (DNN)-based approximation is provided for fast performance evaluation. Analytical results and simulations verify the accuracy of our derived training length expressions and explicitly reveal the impact of system parameters on the average training length. In addition, the modified beam/antenna-domain schemes are shown to have a shorter average training length compared to the basic schemes.

\end{abstract}


% \fi
\IEEEpeerreviewmaketitle

\section{Introduction}

Via exploiting the large number of spatial degrees-of-freedom provided by large-scale antenna arrays, massive MIMO systems can achieve significant performance improvement  compared to conventional MIMO systems \cite{marzetta2010noncooperative, ngo2013energy}. 
One crucial practical issue for massive MIMO downlink is the acquisition of channel state information (CSI) at the base station (BS), especially for frequency-division-duplexing (FDD) systems with no uplink-downlink channel reciprocity \cite{zhang2017sum}. 
Traditional downlink training and channel estimation schemes cause prohibitive training overhead due to the massive number of channel coefficients to be estimated \cite{biguesh2006training}.

Existing studies on the downlink CSI acquisition of massive MIMO can be divided into the following categories. In \cite{adhikary2013joint, choi2014downlink, shen2015joint, zhang2022sum}, the channel statistics, e.g., spatial and/or temporal correlation, are utilized to conduct beamformed channel estimation. 
%\textcolor{blue}{In \cite{zhang2022sum}, a linear precoder is designed to exploit the spatial correlation characteristics of the channels for FDD massive MIMO downlink under a beam-based statistical channel model.}
In \cite{rao2014distributed, gao2015spatially, shen2015joint_tvt, venugopal2017channel, alevizos2018limited, ke2020compressive}, compressive sensing algorithms are designed to exploit the channel sparsity in the angular domain and/or the common sparsity among users and/or subcarriers. 
In \cite{han2019efficient, peng2019downlink, kim2020downlink, yang2019deep_1}, either the partial reciprocity between the uplink and downlink channels, e.g., with similar angle and delay of propagation paths, or their implicit relationships, e.g., both being the functions of user location, are used for channel training designs.

These aforementioned schemes aim to obtain the complete antenna-domain CSI with the smallest possible pilot overhead before the data transmissions, thus the training design and data transmission design are decoupled, which imposes limitations on the tradeoff between training overhead and performance. Further, only the throughput or diversity gain has been considered in these existing works. The quality-of-service (QoS) provided by the obtained CSI is not taken into consideration during the training process. Therefore, the training length or pilot overhead is fixed and does not adjust according to specific channel realizations. For massive MIMO systems, it is possible to use partial CSI to design the beamforming scheme for the data transmission period, especially with the outage probability performance measure. A new training method, namely the interleaved training, is proposed to dynamically adjust the training overhead according to the required QoS and specific channel realizations.

The idea of interleaved training was proposed in \cite{koyuncu2015interleaving,koyuncu2018interleaving} for the downlink of single-user full-digital massive antenna systems with independent and identically distributed (i.i.d.) channels, where the channels of different antennas are trained sequentially and the estimated CSI or indicator is fed back at the end of each training step.
With each feedback, the BS decides to conduct the training of another antenna's channel or to terminate the training process based on whether an outage occurs. Compared to traditional schemes, the interleaved training can achieve a significant reduction in training overhead with no degradation of outage performance.

%{\color{red}Interleaved training can be applied to various communication scenarios as long as involving the required QoS.}
The work in \cite{zhang2018it} applied the idea of interleaved training to beam-domain transmission, where joint beam-based interleaved training and data transmission schemes are proposed for massive MIMO systems with single and multiple users. In \cite{he2018interleaved}, an improved codebook is further designed for interleaved training in millimeter-wave hybrid massive MIMO downlink. 
In \cite{zhang2020interleaved}, a joint interleaved training and transmission design is proposed for large-intelligent surface (LIS) assisted systems under i.i.d. Rayleigh fading channels. 
Recently, an interleaved training design is proposed for multi-user massive MIMO downlink in \cite{jing2022sinr,jing2023interleaved}, in which analytical results on the training length and the transmission success rate are provided for the maximum-ratio transmission (MRT) precoding. 
Different from the single-user scheme, the multi-user scheme needs to judge whether the signal-to-interference-plus-noise (SINR) requirements of all users can be satisfied with partial CSI during the interleaved training procedure.
The advantages of interleaved training are clearly demonstrated in the above studies.

Different from previous papers on interleaved training \cite{koyuncu2015interleaving,koyuncu2018interleaving,zhang2018it,he2018interleaved,zhang2020interleaved, jing2022sinr,jing2023interleaved}, in this paper, we focus on exploiting channel statistics to further improve the performance of interleaved training for single-user massive MIMO downlink systems. Both beam-domain and antenna-domain modified interleaved training designs are proposed, and we further provide analytical results on the average training length of the proposed schemes and their comparison with those of the basic schemes. Detailed contributions are summarized as follows.

\begin{itemize}
\item In the modified beam-domain scheme, both the beam direction and the beam training order are optimized based on the channel correlation information. In the modified antenna-domain scheme, the BS antenna training order is dynamically adjusted during the training process according to the channel values of the already trained antennas. To this end, we derive the conditional distribution of the untrained BS channels given the values of the trained channels under general correlated channels.  And for exponentially correlated channels, we demonstrate that the conditional distribution of any untrained antenna's channel is only dependent on the channels of its nearest antennas on both sides in the already-trained antenna set, which significantly simplifies the complexity of the modified antenna-domain scheme.
\item  Closed-form expressions of the average training length are derived for the basic and modified beam-domain interleaved training schemes with general correlated channels. For exponentially correlated channels, we further provide a simplified approximation of the average training length for the modified scheme when the number of BS antennas is large. 
\item A closed-form average-training-length expression is also derived for the basic antenna-domain interleaved training with general correlated channels, and its simple approximation is given for exponentially correlated channels. For the average training length of the modified antenna-domain interleaved training, we propose a deep neural network (DNN)-based approximation to achieve fast performance evaluation.
%\item Based on the derivations for above analytical expressions of average training length, we provide theoretical analysis on the comparison between basic and modified beam/antenna-domain interleaved training schemes, and reveal the impact of important system parameters, e.g., channel correlation, antenna number, the requirement of signal-noise-ratio (SNR).
\item Simulation results verify our derived analytical expressions and theoretical analysis on the impact of systems parameters, e.g., channel correlation, antenna number, and the requirement of the signal-noise-ratio (SNR), on the average training length of the basic and modified antenna/beam-domain interleaved training schemes under two typical correlated channels. In addition, simulations also demonstrate that our proposed modified antenna and beam-domain interleaved training schemes both outperform the basic interleaved training schemes.
\end{itemize}

The remainder of this paper is organized as follows.
In Section \ref{System Model}, we introduce the single-user massive MIMO downlink system with two typical channel correlation models and the basic antenna/beam-domain interleaved training scheme.
In Section \ref{Modified Beam-domain Interleaved Training}, the modified beam-domain interleaved training is proposed along with related analytical results.
In Section \ref{Modified Antenna-domain Interleaved Training}, we propose the modified antenna-domain interleaved training and conduct theoretical analysis.
Simulations are provided in Section \ref{Simulation}.
Section \ref{Conclusion} summarizes this work. Some proofs are included in the appendix.

\emph{Notation}: Bold upper and bold lower case letters denote matrices and vectors. 
$\mathbb{{C}}^{m \times n}$ denotes the $m$ by $n$ dimensional complex space.
$\mathbf{I}_{n}$ denote the $n$-dimensional identity matrix. 
The conjugate transpose, transpose, determinant, adjugate matrix, rank and inverse of $\mathbf{A}$ are denoted by $\mathbf{A}^{\rm{H}}$, $\mathbf{A}^{{\rm T}}$, ${\rm det}\left( \mathbf{A} \right)$, ${\rm adj}\left( \mathbf{A} \right)$, ${\rm rank}\left( \mathbf{A} \right)$ and $\mathbf{A}^{\rm -1}$. 
The vector $\mathbf{a}_i$ denotes the $i$-th column of the matrix $\mathbf{A}$.
For a vector ${\mathbf a}$, ${a}_n$ is the $n$-th element of $\mathbf{a}$ and $\mathbf{a}_{\mathbb{S}}$ is the sub-vector composed of the $s$-th element of $\mathbf{a}$ for $s\in\mathbb{S}$ when the subscript is an index set $\mathbb{S}$. Similarly, $[\mathbf{A}]_{m,n}$ is the $(m,n)$-th element of $\mathbf{A}$ and $[\mathbf{A}]_{\mathbb{S},\mathbb{T}}$ is the sub-matrix composed of the $(s,t)$-th element of $\mathbf{A}$ for $s\in\mathbb{ S} $ and $t\in\mathbb{ T}$. Define $[m:n]$ as the set $\{m,m+1,\dots,n\}$.
$\Pr \left(\mathcal{A}\right)$ represents the probability of event $\mathcal{A}$.
$\left\lfloor  \right\rfloor $ represents the floor function.
$\left\Vert \mathbf{a}\right\Vert $ denotes the Euclidean norm of $\mathbf{a}$.
${\rm diag}(\mathbf{a})$ is the diagonal matrix whose diagonal entries are elements of vector $\mathbf{a}$. 
$f_{X}\left(\cdot\right)$ denotes the probability density function (PDF) of a random variable (RV) $X$.
$\mathcal{CN}\left({\boldsymbol{\mu}},\boldsymbol{\Sigma}\right)$ denotes the circularly symmetric complex Gaussian distribution with mean vector $\boldsymbol{\mu}$ and covariance matrix $\boldsymbol{\Sigma}$.
${{\chi }^{2}}\left(k\right)$ denotes the chi-squared distribution with $k$ being the degrees of freedom. ${{\chi }^{2}}\left(k,\lambda\right)$ denotes the noncentral chi-squared distribution with $k$ and $\lambda$ being the degrees of freedom and the non-centrality parameter.
${{Q}_{1}}\left(a,b \right)$ is the first order Marcum Q-function.
$\cong$ denotes the equality in distribution.

\section{System Model}
\label{System Model}

We consider a massive MIMO downlink system with an $M$-antenna BS and a single-antenna user equipment (UE). The downlink BS-UE channel, denoted as $\mathbf{h}$, is modeled as a circular-symmetric complex Gaussian vector following the distribution $\mathcal{CN}\left(\mathbf{0},\mathbf{R}_{\mathbf{h}}\right)$, where $\mathbf{R}_{\mathbf{h}}$ is the channel covariance matrix.
One typical correlation model is the one-ring correlation model \cite{adhikary2013joint},\cite{you2015pilot}, which is expressed as
\setlength\abovedisplayskip{1pt}
\setlength\belowdisplayskip{1pt}
\begin{equation}\label{eq_onering}
\mathbf{R}_{\mathbf{h}}=\int_{\Theta_{\rm min}}^{\Theta_{\rm max}}g(\theta)\boldsymbol{\alpha}(\theta)\boldsymbol{\alpha}^{{\rm H}}(\theta)d\theta,
\end{equation}
where $\left[\Theta_{\rm min}, \Theta_{\rm max}\right]$ is the angle interval of the channel power seen at the BS, $g(\cdot)$ represents the power angle spectrum (PAS), satisfying $\int_{\Theta_{\rm min}}^{\Theta_{\rm max}}g(\theta)d\theta=1$, and $\boldsymbol{\alpha}(\theta)\in\mathbb{C}^{M\times1}$ is the BS array response vector. For the uniform linear array (ULA), $\boldsymbol{\alpha}(\theta)=\left[1,\dots,e^{-j2\pi D\sin\left(\theta\right)\left(M-1\right)}\right]^{{\rm T}}$ where $D$ is the antenna spacing ratio.
Another typical correlation model is the exponential one \cite{loyka2001channel}, i.e., 
\setlength\abovedisplayskip{1pt}
\setlength\belowdisplayskip{1pt}
\begin{equation}\label{eq:}
	[\mathbf{R}_{\mathbf{h}}]_{m,n}=\rho^{m-n},\forall m \ge n, m,n=1,...,M,
\end{equation}
where $\rho$, satisfying $r = \left| \rho  \right|<1$, is the channel correlation between adjacent antennas. 
This is a simple single-parameter model commonly used for many communication problems, which is also physically reasonable in the sense that the correlation decreases with increasing distance between antennas, e.g., in the ULA.

The downlink transmission can be represented as
\setlength\abovedisplayskip{1pt}
\setlength\belowdisplayskip{1pt}
\begin{equation}\label{eq3}
y=\sqrt{P}\mathbf{h}^{\rm H}\mathbf{w}s+n,
\end{equation}
where $y$ is the received signal at the user, $\mathbf{w} \in \mathbb{C}^{M}$ is the antenna-domain beamformer at the BS with the unit norm, i.e., $\left\Vert \mathbf{w}\right\Vert =1$,
$s$ is the transmitted symbol with unit average power,
$P$ is the transmit power and $n$ is the normalized receive noise at the UE which follows $\mathcal{CN}\left(0,1\right).$ 
The received SNR can be written as
\setlength\abovedisplayskip{1pt}
\setlength\belowdisplayskip{1pt}
\begin{equation}\label{SNR_eq4}
{\rm SNR}=P\left|\mathbf{h}^{\rm H}\mathbf{w}\right|^{2}.
\end{equation}
If the beam-domain transmission is conducted, 
$\mathbf{w}$ can be decomposed into two parts: $\mathbf{w} = \mathbf{W}_{\rm O}\mathbf{w}_{\rm I}$, where $\mathbf{W}_{\rm O} \in \mathbb{C}^{M\times B}$ ($B\le M$) is the external beamforming matrix, $\mathbf{w}_{\rm I} \in \mathbb{C}^{B}$ is the beam-domain beamformer, and $B$ is the number of beams.
One typical $\mathbf{W}_{\rm O}$ is the normalized discrete Fourier transformation (DFT) matrix $\mathbf{D}\in \mathbb{C}^{M\times M}$ with $[\mathbf{D}]_{m,n} = e^{j2\pi\frac{(m-1)(n-1)}{M}}/\sqrt{M}$. 
Define the $B$-dimensional beam-domain channel as $\bar{\mathbf{h}}= \mathbf{W}_{\rm O}^{\rm H}\mathbf{h}$. We have $\bar{\mathbf{h}}\sim \mathcal{CN}(\mathbf{0}, \mathbf{R}_{\bar{\mathbf{h}}})$ with $\mathbf{R}_{\bar{\mathbf{h}}} = \mathbf{W}_{\rm O}^{\rm H} \mathbf{R}_{\mathbf{h}}\mathbf{W}_{\rm O}$. Eq. \eqref{SNR_eq4} can be converted to
\setlength\abovedisplayskip{1pt}
\setlength\belowdisplayskip{1pt}
\begin{equation}\label{SNR_eq5}
{\rm SNR}=P\left|\bar{\mathbf{h}}^{\rm H}\mathbf{w}_{\rm I}\right|^{2}.
\end{equation} 
For a given target data transmission rate $R_{\rm th}$, an outage event occurs if ${\rm log_{2}\left(1+SNR\right)}< R_{\rm th}$, or equivalently if ${\rm SNR}<P\alpha_{\rm th}$, where $\alpha_{\rm th}=(2^{R_{\rm th}}-1)/P$ is the normalized receive SNR threshold.

\subsection{General Framework of Interleaved Training and the Basic Training Scheme}

In this subsection, we introduce the general antenna-domain and beam-domain interleaved training and give the basic interleaved training algorithms proposed in \cite{koyuncu2018interleaving, zhang2018it} as the baseline of our study. For a uniform representation, we define 
\setlength\abovedisplayskip{1pt}
\setlength\belowdisplayskip{1pt}
\begin{equation}
	\left(\tilde{\mathbf{h}},\tilde{\mathbf{w}}, L \right)=\begin{cases} \left(\mathbf{h}, \mathbf{w}, M\right), \text{ for antenna-domain training}
\\  
\left(\bar{\mathbf{h}},\mathbf{w}_{\rm I}, B \right), \text{ for beam-domain training}\\
\end{cases}.
\end{equation}
In the general antenna/beam-domain interleaved training scheme, the BS trains the channel of one antenna/beam for each step, and the order of the antennas/beams during the training is determined according to a predefined criterion.
After the $l$-th training step, the UE knows $\tilde{\mathbf{h}}_{\mathbb{A}_{l}}$ with $\mathbb{A}_{l}$ denoting the set of indices of the already trained BS antennas/beams \footnote{{The main purpose of interleaved training is to reduce the training overhead. In order to focus on the theoretical analysis and give more insights, we do not consider the error resulting from channel estimation and feedback quantization in our study.}}. To maximize the receive SNR based on this currently acquired CSI, the BS can conduct the following downlink beamforming  
\setlength\abovedisplayskip{1pt}
\setlength\belowdisplayskip{1pt}
\begin{equation}\label{eq7}
	\tilde{w}_n=\begin{cases}\frac{\tilde{h}_n}{\left\Vert\tilde{\mathbf{h}}_{\mathbb{A}_l}\right\Vert}, &\text{if    } n\in{\mathbb{A}_l}
\\  
0, &\text{if   } n\notin{\mathbb{A}_l}\\
\end{cases}.
\end{equation}
The receive SNR of this beamformer is thus ${\rm SNR}=P\left\Vert \tilde{\mathbf{h}}_{\mathbb{A}_l}\right\Vert ^{2}$. Based on whether an outage occurs, i.e., $\left\Vert \tilde{\mathbf{h}}_{\mathbb{A}_l}\right\Vert ^{2}<\alpha_{\rm th}$, the UE decides to notify the BS to continue training by one bit $0$, or feed back one bit $1$ and channel values of the already trained antennas/beams to the BS for transmission beamforming.

In the basic antenna/beam-domain interleaved training scheme as shown in Algorithm \ref{algorithm_1}, the channel of one BS antenna or one DFT beam is trained for each step, and the antennas/beams are trained sequentially following their original indices, i.e., after the $l$-th training step, the index set of the already trained antennas/beams is $[1:l]$.

\begin{algorithm}[htp]
%\begin{spacing}{1.2} 
%\setstretch{1.2} 
\linespread{1} \selectfont
  \caption{Basic antenna/beam-domain interleaved training scheme\cite{koyuncu2018interleaving, zhang2018it}} 
    \begin{algorithmic}[1]
\item \textbf{Initialization:} $\mathbb{A}_{1}=\left\{ 1\right\} $; $l=1$; The BS sends a pilot for the UE to acquire $\tilde{h}_1$;
\item \textbf{While} $\left\Vert \tilde{\mathbf{h}}_{\mathbb{A}_l}\right\Vert ^{2}<\alpha_{\rm th}$ \&
$l<L$ \textbf{do}
\item \hspace{0.5cm}The UE sends one bit $0$ to the BS;
\item \hspace{0.5cm}The BS sends a pilot for the UE to acquire $\tilde{h}_{l+1}$;
\item \hspace{0.5cm}$l=l+1$; $\mathbb{A}_l=\left\{\mathbb{A}_{l-1},l\right\} $; 
\item \textbf{end}
\item \textbf{if} $\left\Vert \tilde{\mathbf{h}}_{\mathbb{A}_l}\right\Vert ^{2}\geq\alpha_{\rm th}$
\item \hspace{0.5cm}The UE feeds back one bit $1$ and $\tilde{\mathbf{h}}_{\mathbb{A}_l}$ to the BS;
\item \hspace{0.5cm}The BS conducts downlink beamforming according to Eq. \eqref{eq7};
\item \textbf{else}
\item \hspace{0.5cm}The UE feeds back one bit $0$ to the BS;
\item \textbf{end}
    \end{algorithmic} \label{algorithm_1}  
%\end{spacing}
\end{algorithm}

\section{Modified Beam-Domain Interleaved Training and Performance Analysis}\label{Modified Beam-domain Interleaved Training}

In the basic beam-domain interleaved training scheme \cite{zhang2018it}, the adopted DFT beams give no guarantee to align the effective propagation paths, and nor does it consider setting a higher training priority for beams with stronger average power. In the following, we explore the channel covariance matrix to improve the beam-domain interleaved training via addressing the above issues. In addition, we perform analysis on the average training length of the modified beam-domain interleaved training and compare it with that of the basic beam-domain interleaved training to reveal the advantages of the modified design. The methods of acquiring channel covariance matrix at the BS in massive MIMO systems can be referred to \cite{decurninge2015channel, neumann2018covariance, li2022downlink}.

\subsection{Modified Beam-Domain Training Design}
Recall that the channel covariance matrix $\mathbf{R}_{\mathbf{h}}$ is positive semi-definite and we denote its rank as $r_M$. 
We consider the compact eigenvalue decomposition of $\mathbf{R}_{\mathbf{h}}$: $\mathbf{R}_{\mathbf{h}} = \mathbf{U}\boldsymbol{\Sigma}\mathbf{U}^{\rm H}$, where $\mathbf{U}$ is an $M\times r_M$ semi-unitary matrix and $\boldsymbol{\Sigma}={\rm diag}\{\delta_1,...,\delta_{r_M}\}$ with ${\delta }_{1}\ge{\delta }_{2}\cdots \ge {\delta }_{r_M}>0$. 
With the knowledge of $\mathbf{U}$ and $\boldsymbol{\Sigma}$, the BS can set $\mathbf{W}_{\rm O} = \mathbf{U}$, implying that $B = r_M$, and  therefore $\bar{\mathbf{h}}=\mathbf{U}^{\rm H}\mathbf{h}$ is the $B$-dimentional vector of the beam-domain channel coefficients. In the modified scheme, the BS trains the $B$ effective beams $\mathbf{u}_1,\mathbf{u}_2,\cdots,\mathbf{u}_{B}$ in turn, such that the beams are trained with decreasing average power. 
After $b$ steps of beam training, the BS obtains the beam-domain channels $\bar{\mathbf{h}}_{\mathbb{A}_{b}}$ where $\mathbb{A}_{b}=\{1,...,b\}$. The BS conducts the beam-domain precoding $\mathbf{w}_{\rm I}\in\mathbb{C}^B$ according to Eq. \eqref{eq7}.
And an outage occurs if $ 
	\left\Vert \bar{\mathbf{h}}_{\mathbb{A}_{b}} \right\Vert ^2 < {{\alpha }_{{\rm th}}}$. 
With this beam ordering, the specific process of the modified beam-domain training scheme can be referred to as Algorithm \ref{algorithm_1}.
The modified scheme has both beam alignment and ordering through the eigenmatrix $\mathbf{U}$ of the channel covariance matrix. This is the major difference from the basic one. 

From the Toeplitz eigen-subspace approximation result in \cite{grenander1958toeplitz},  the eigenvectors of the one-ring covariance matrix in Eq. \eqref{eq_onering} and those of the exponential covariance matrix in Eq. \eqref{eq:} can both be well approximated by the columns of a DFT matrix for $M\gg 1$. As $M$ increases asymptotically to infinity, both the modified scheme and the basic scheme use the DFT codebook for beam training and their difference then only lies in the order of the beams during training. In this case, since the modified scheme trains the beams with decreasing average power sequentially, it has a shorter average training length.

\subsection{Average Training Length Analysis}\label{Average Training Length Analysis}

Considering that the difference between the basic beam-domain interleaved training and the modified one lies only in the use of the training beam codebook, we first give the analysis of the average training length for the beam-domain interleaved training scheme with any given beam codebook, based on which the average training length of the modified scheme and its comparison with that of the basic scheme are subsequently given. 

\subsubsection{Analysis for the General Beam-Domain Scheme}
Recall that the receiver SNR after the $b$-th training step is ${\rm SNR}=P\left\Vert \bar{\mathbf{h}}_{\mathbb{A}_{b}} \right\Vert ^{2}$. From Algorithm \ref{algorithm_1}, we can see that the training stops after the $b$-th training step with probability $\Pr \left( {{\left| {{\bar{\mathbf{h}}}_{\mathbb{A}_{1}}} \right|}^{2}}\geq\alpha_{\rm th}  \right)$ for $b=1$ and $\Pr \left( {{\left\| {\bar{\mathbf{h}}_{\mathbb{A}_{b}}} \right\|}^{2}}\geq\alpha_{\rm th} \right)-\Pr \left( {{\left\| {\bar{\mathbf{h}}_{\mathbb{A}_{b-1}}} \right\|}^{2}}\geq\alpha_{\rm th}  \right)$ for $b=2,...,B-1$. And the training stops after the $B$-th training with probability $1-\Pr \left( {{\left\| {\bar{\mathbf{h}}_{\mathbb{A}_{B-1}}} \right\|}^{2}}\geq\alpha_{\rm th}  \right)$.
 The average training length of the beam-domain interleaved training scheme can be expressed as 
\setlength\abovedisplayskip{1pt}
\setlength\belowdisplayskip{1pt}
\begin{equation}\label{eq11}
%\begin{aligned}
%	L_t  &=\Pr \left( {{\left| {{\bar{\mathbf{h}}}_{\mathbb{A}_{1}}} \right|}^{2}}\geq\alpha_{\rm th}  \right)+\sum_{b=2}^{B-1}b\left[ \Pr \left( {{\left\| {\bar{\mathbf{h}}_{\mathbb{A}_{b}}} \right\|}^{2}}\geq\alpha_{\rm th} \right)\right.\\ 
%&\left.-\Pr \left( {{\left\| {\bar{\mathbf{h}}_{\mathbb{A}_{b-1}}} \right\|}^{2}}\geq\alpha_{\rm th}  \right) \right]\\
%&+B \left[ 1-\Pr \left( {{\left\| {\bar{\mathbf{h}}_{\mathbb{A}_{B-1}}} \right\|}^{2}}\geq\alpha_{\rm th}  \right) \right]\\
%&	=1+\sum\limits_{b=1}^{B-1}{\Pr \left( {{\left\| {\bar{\mathbf{h}}_{\mathbb{A}_{b}}} \right\|}^{2}}<\alpha_{\rm th}  \right)}.
%\end{aligned}
\begin{aligned}
\hspace{-0.2cm}L_t  &=\Pr \left( {{\left| {{\bar{\mathbf{h}}}_{\mathbb{A}_{1}}} \right|}^{2}}\geq\alpha_{\rm th}  \right)+\\ 
&\sum_{b=2}^{B-1}b\left[ \Pr \left( {{\left\| {\bar{\mathbf{h}}_{\mathbb{A}_{b}}} \right\|}^{2}}\geq\alpha_{\rm th} \right)-\Pr \left( {{\left\| {\bar{\mathbf{h}}_{\mathbb{A}_{b-1}}} \right\|}^{2}}\geq\alpha_{\rm th}  \right) \right]\\
&+B \left[ 1-\Pr \left( {{\left\| {\bar{\mathbf{h}}_{\mathbb{A}_{B-1}}} \right\|}^{2}}\geq\alpha_{\rm th}  \right) \right]\\
&	=1+\sum\limits_{b=1}^{B-1}{\Pr \left( {{\left\| {\bar{\mathbf{h}}_{\mathbb{A}_{b}}} \right\|}^{2}}<\alpha_{\rm th}  \right)}.
\end{aligned}
\end{equation}

\begin{theorem}\label{Theorem1}
For an arbitrary beam codebook $\mathbf{W}_{\rm O}$, recall that $\mathbf{R}_{\bar{\mathbf{h}}} = \mathbf{W}_{\rm O}^{\rm H} \mathbf{R}_{\mathbf{h}}\mathbf{W}_{\rm O}$ and  define $\tilde{\mathbf{R}}_b = \left[\mathbf{R}_{\bar{\mathbf{h}}}^{1/2}\right]_{\mathbb{A}_{b},[1:M]}$. Consider the compact eigenvalue decomposition $\tilde{\mathbf{R}}_b^{{\rm H}}\tilde{\mathbf{R}}_b={{\mathbf{U}}_{b}} {{\mathbf{\Sigma }}_{b}} \mathbf{U}_{b}^{{\rm H}}$ where $\mathbf{U}_b$ is an $M\times r_b$ semi-unitary matrix, ${{\mathbf{\Sigma }}_{b}} ={\rm  diag}\left\{ {{\delta }_{b,1}},\dots ,{{\delta }_{b,r_b}}\right\}$ and $r_b$ is the rank of $\tilde{\mathbf{R}}_b^{{\rm H}}\tilde{\mathbf{R}}_b$.
Suppose that there are $T_b$ different eigenvalues with value of $\bar{\delta}_{b,t}$ and repeated time of $r_{b,t}$ for $t=1,...,T_b$.
Define $\mathbf{r}_b=[{{r}_{b,1}},...,{{r}_{b,T_b}}]^{\rm T}$.
The average training length of the general beam-domain interleaved training scheme under correlated channels can be expressed as 
\setlength\abovedisplayskip{1pt}
\setlength\belowdisplayskip{1pt}
\begin{equation}\label{eqtl}
\begin{aligned}
\hspace{-0.1cm}L_t=1+&\sum\limits_{b=1}^{B-1}{\prod\limits_{t=1}^{T_b}\left(\frac{1}{{\bar{\delta }_{b,t}}}\right)^{{{r}_{b,t}}}}\sum\limits_{k=1}^{T_b}\sum\limits_{s=1}^{{{r}_{b,k}}}{{\left(-1 \right)}^{{{r}_{b,k}}-s}}{{\bar{\delta }_{b,k}}}^{{{r}_{b,k}}-s+1}\\
&\times{{\Psi }_{b,k,s,\mathbf{r}_b}}\left[ 1-{{e}^{-\frac{\alpha_{\rm th}}{{\bar{\delta }_{b,k}}}}}\sum\limits_{u=0}^{{{r}_{b,k}}-s}{\frac{{{\left(\frac{\alpha_{\rm th}}{{\bar{\delta }_{b,k}}}  \right)}^{u}}}{u!}} \right],
\end{aligned}
\end{equation}
where ${{\Psi }_{b,k,s,\mathbf{r}_b}}={{\left( -1 \right)}^{{{r}_{b,k}}-1}}\sum\limits_{\mathbf{i}\in {{\Omega }_{b,k,s}}}\prod\limits_{n\ne k}\left( \begin{matrix}
{{i}_{n}}+{{r}_{b,n}}-1 \\
{{i}_{n}} \\
\end{matrix} \right)$ $\times{{\left( \frac{1}{{\bar{\delta }_{b,n}}}-\frac{1}{{\bar{\delta }_{b,k}}} \right)}^{-\left( {{i}_{n}}+{{r}_{b,n}} \right)}}$, $\mathbf{i}={{\left[ {{i}_{1}},\ldots ,{{i}_{T_b}} \right]}^{\rm T}}$
 and ${{\Omega }_{b,k,s}}=\left\{ \left[ {{i}_{1}},\ldots ,{{i}_{T_b}} \right]\in {{\mathbb{Z}}^{T_b}};\sum\limits_{j=1}^{T_b}{{{i}_{j}}=s-1,{{i}_{k}}=0,{{i}_{j}}\ge 0} \text{ for all } j \right\}.$
\end{theorem}
\begin{proof} 
Recall that $\bar{\mathbf{h}}= \mathbf{W}_{\rm O}^{\rm H}\mathbf{h}\sim \mathcal{CN}(\mathbf{0}, \mathbf{R}_{\bar{\mathbf{h}}})$.
We have $\bar{\mathbf{h}}_{\mathbb{A}_{b}} \cong \tilde{\mathbf{R}}_b\mathbf{h}_{\rm iid}$ with ${\mathbf{h}}_{\rm iid}\sim \mathcal{C}\mathcal{N}\left( \mathbf{0},{{\mathbf{I}}_{M}} \right)$ and $
{{\left\| \bar{\mathbf{h}}_{\mathbb{A}_{b}} \right\|}^{2}}=\mathbf{h}_{\rm iid}^{{\rm H}}{{\mathbf{U}}_{b}} {{\mathbf{\Sigma }}_{b}} \mathbf{U}_{b}^{{\rm H}}\mathbf{h}_{\rm iid}\cong{{ \tilde{\mathbf{{\rm h}}}_{\rm iid}^{b,{\rm H}} }}{{\mathbf{\Sigma }}_{b}}\tilde{\mathbf{h}}_{\rm iid}^{b}
$
with $\tilde{\mathbf{h}}_{\rm iid}^{b}\sim \mathcal{C}\mathcal{N}\left( \mathbf{0},{{\mathbf{I}}_{r_b}} \right)$.
Therefore,
\setlength\abovedisplayskip{1pt}
\setlength\belowdisplayskip{1pt}
\begin{equation}\label{eq15}
{{\left\| \bar{\mathbf{h}}_{\mathbb{A}_{b}} \right\|}^{2}}=\sum\limits_{j=1}^{r_b}{{{\delta }_{b,j}}{{\left| \tilde{h}_{{\rm iid},j}^{b} \right|}^{2}}} \cong \sum\limits_{t=1}^{T_b}{\frac{1}{2}\bar{\delta}_{b,t}} {{Q}_{b,t}},
\end{equation}
where 
${{Q}_{b,t}}\sim {{\chi }^{2}}\left( 2{{r}_{b,t}} \right)$.
By using results in \cite{bjornson2009exploiting} on the sum of independent chi-square random variables, the PDF of ${{\left\| \bar{\mathbf{h}}_{\mathbb{A}_{b}} \right\|}^{2}}$ is
\setlength\abovedisplayskip{1pt}
\setlength\belowdisplayskip{1pt}
\begin{equation}\label{eq16}
\begin{aligned}
f\left( {{{\left\| \bar{\mathbf{h}}_{\mathbb{A}_{b}} \right\|}^{2}} }=x ;\mathbf{r}_b,{\bar{\delta}_{b,1}},\ldots,{\bar{\delta}_{b,T_b}} \right)=\prod\limits_{t=1}^{T_b}{\left(\frac{1}{{\bar{\delta }_{b,t}}}\right)^{{{r}_{b,t}}}}\\
\times\sum\limits_{k=1}^{T_b}{\sum\limits_{s=1}^{{{r}_{b,k}}}{\frac{{{\Psi }_{b,k,s,\mathbf{r}_b}}}{\left( {{r}_{b,k}}-s \right)!}}}{{\left( -x \right)}^{{{r}_{b,k}}-s}}{{e}^{-\frac{x}{{\bar{\delta }_{b,k}}}}}.
\end{aligned}
\end{equation} 
Therefore, we have
\setlength\abovedisplayskip{1pt}
\setlength\belowdisplayskip{1pt}
\begin{equation}\label{eq17}
\begin{aligned}
&\hspace{-0.5cm}\Pr\left( {{{\left\| \bar{\mathbf{h}}_{\mathbb{A}_{b}} \right\|}^{2}}}<\alpha_{\rm th}\right)=\prod\limits_{t=1}^{T_b}{\left(\frac{1}{{\bar{\delta }_{b,t}}}\right)^{{{r}_{b,t}}}}\times\\
&\hspace{0.5cm} \sum\limits_{k=1}^{T_b}{\sum\limits_{s=1}^{{{r}_{b,k}}}{\frac{{{\Psi }_{b,k,s,\mathbf{r}_b}}}{\left( {{r}_{b,k}}-s \right)!}}}\int_{0}^{\alpha_{\rm th}}{{{\left( -x \right)}^{{{r}_{b,k}}-s}}{{e}^{-\frac{x}{{\bar{\delta }_{b,k}}}}}dx} \\ 
&\hspace{-0.5cm} =\prod\limits_{t=1}^{T_b}{\left(\frac{1}{{\bar{\delta }_{b,t}}}\right)^{{{r}_{b,t}}}}\sum\limits_{k=1}^{T_b}{\sum\limits_{s=1}^{{{r}_{b,k}}}{{{{\left( -1 \right)}^{{{r}_{b,k}}-s}}{\bar{\delta }_{b,k}}^{{{r}_{b,k}}-s+1}{{\Psi }_{b,k,s,\mathbf{r}_b}}}}}\\
& \hspace{7em}\times\left[ 1-{{e}^{-\frac{\alpha_{\rm th}}{{\bar{\delta }_{b,k}}}}}\sum\limits_{u=0}^{{{r}_{b,k}}-s}{\frac{{{\left(\frac{\alpha_{\rm th}}{{\bar{\delta }_{b,k}}}  \right)}^{u}}}{u!}} \right]. \\ 
\end{aligned}
%\begin{aligned}
%&\Pr\left( {{{\left\| \bar{\mathbf{h}}_{\mathbb{A}_{b}} \right\|}^{2}}}<\alpha_{\rm th}\right)\\
%&=\prod\limits_{t=1}^{T_b}{\left(\frac{1}{{\bar{\delta }_{b,t}}}\right)^{{{r}_{b,t}}}}\sum\limits_{k=1}^{T_b}{\sum\limits_{s=1}^{{{r}_{b,k}}}{\frac{{{\Psi }_{b,k,s,\mathbf{r}_b}}}{\left( {{r}_{b,k}}-s \right)!}}}\int_{0}^{\alpha_{\rm th}}{{{\left( -x \right)}^{{{r}_{b,k}}-s}}{{e}^{-\frac{x}{{\bar{\delta }_{b,k}}}}}dx} \\ 
%& =\prod\limits_{t=1}^{T_b}{\left(\frac{1}{{\bar{\delta }_{b,t}}}\right)^{{{r}_{b,t}}}}\sum\limits_{k=1}^{T_b}{\sum\limits_{s=1}^{{{r}_{b,k}}}{{{{\left( -1 \right)}^{{{r}_{b,k}}-s}}{\bar{\delta }_{b,k}}^{{{r}_{b,k}}-s+1}{{\Psi }_{b,k,s,\mathbf{r}_b}}}}}\\
%& \hspace{9em}\times\left[ 1-{{e}^{-\frac{\alpha_{\rm th}}{{\bar{\delta }_{b,k}}}}}\sum\limits_{u=0}^{{{r}_{b,k}}-s}{\frac{{{\left(\frac{\alpha_{\rm th}}{{\bar{\delta }_{b,k}}}  \right)}^{u}}}{u!}} \right]. \\ 
%\end{aligned}
\end{equation}
Via substituting Eq. \eqref{eq17} into Eq. \eqref{eq11}, Eq. \eqref{eqtl} can be obtained.
\end{proof}

%For the special case with exponential channel covariance, a simplified approximation of the average training length for the antenna-domain interleaved training scheme is given as follows.

\subsubsection{Analysis for the Modified Beam-Domain Scheme}

The derivation for the average training length of the modified beam-domain interleaved training scheme can be referred to Theorem \ref{Theorem1} by setting the beam codebook as $\mathbf{W}_{\rm O}=\mathbf{U}$. Therefore, $\mathbf{R}_{\bar{\mathbf{h}}} = \mathbf{U}^{\rm H} \mathbf{R}_{\mathbf{h}}\mathbf{U}=\boldsymbol{\Sigma}$ and $\tilde{\mathbf{R}}_b = \left[\boldsymbol{\Sigma}^{\frac{1}{2}}\right]_{[1:b],[1:M]}$, which leads to ${{\delta }_{b,j}}={{\delta }_{j}}$ for $b=1,\dots,B,j=1,...,b$. It is noteworthy that since the modified scheme uses the eigenvectors of the channel covariance matrix as the beam codebook, $\delta_{b,j}$ is independent of the training step index $b$.
Suppose that there are $\bar{T}_b$ different values in the first $b$ eigenvalues of $\mathbf{R}_{{\mathbf{h}}}$ ${{\delta }_{j}}, j=1,...,b$ with value of $\bar{\delta}_{b,t}$ and repeated times $\bar{r}_{b,t}$ for $t=1,...,\bar{T}_b$. Define $\bar{\mathbf{r}}_b=[{\bar{r}_{b,1}},...,{\bar{r}_{b,\bar{T}_b}}]^{\rm T}$.
\begin{corollary}\label{Ltmodifed_beamdomain}
	 The average training length of the modified beam-domain interleaved training scheme can be expressed as:
\setlength\abovedisplayskip{1pt}
\setlength\belowdisplayskip{1pt}
\begin{equation}\label{eqbeam}
\begin{aligned}
\hspace{-0.3cm}	L_t= 1+&\sum\limits_{b=1}^{B-1}\prod\limits_{t=1}^{\bar{T}_b}{\left(\frac{1}{\bar{\delta }_{b,t}}\right)^{{\bar{r}_{b,t}}}}\hspace{-0.1cm}\sum\limits_{k=1}^{\bar{T}_b}\sum\limits_{s=1}^{\bar{r}_{b,k}}}{{{{\left( -1 \right)}^{{\bar{r}_{b,k}}-s}}\bar{\delta}_{k}}^{{\bar{r}_{b,k}}-s+1}\\
&\times{{\Psi }_{b,k,s,\bar{\mathbf{r}}_b}}\left[ 1-{{e}^{-\frac{\alpha_{\rm th}}{{\bar{\delta}_{b,k}}}}}\sum\limits_{u=0}^{{\bar{r}_{b,k}}-s}{\frac{{{\left( {\frac{\alpha_{\rm th}}{\bar{\delta }_{b,k}}} \right)}^{u}}}{u!}} \right],
\end{aligned}
\end{equation}
where ${{\Psi }_{b,k,s,\bar{\mathbf{r}}_b}}={{\left( -1 \right)}^{{\bar{r}_{b,k}}-1}}\sum\limits_{\mathbf{i}\in {{\Omega }_{b,k,s}}}\prod\limits_{n\ne k}\left( \begin{matrix}
			{{i}_{n}}+{\bar{r}_{b,n}}-1  \\
			{{i}_{n}}  \\
		\end{matrix} \right)$ $\times{{\left( \frac{1}{\bar{\delta }_{b,n}}-\frac{1}{\bar{\delta }_{b,k}} \right)}^{-\left( {{i}_{n}}+{\bar{r}_{b,n}} \right)}}$, $\mathbf{i}={{\left[ {{i}_{1}},\ldots ,{{i}_{\bar{T}_b}} \right]}^{\rm T}}$ and ${{\Omega}_{b,k,s}}=\left\{\left[{i}_{1},\ldots,{i}_{\bar{T}_b}\right]\in \mathbb{Z}^{\bar{T}_b};\sum\limits_{j=1}^{\bar{T}_b}{{{i}_{j}}=s-1,{{i}_{k}}=0,{{i}_{j}}\ge 0} \text{ for all } j\right\}$.
\end{corollary}
\begin{proof}
The result can be directly obtained from Theorem \ref{Theorem1}.
\end{proof}

\begin{corollary}\label{Ltmodifed_beam_appro}
	For channels with the exponential covariance matrix in Eq. \eqref{eq:} and $0 \le  r < 1$, a large $M$ approximation of the average training length for the modified beam-domain interleaved training scheme can be written as
\setlength\abovedisplayskip{1pt}
\setlength\belowdisplayskip{1pt}
\begin{equation}\label{eq29}
	L_t =  1+\sum\limits_{m=1}^{M-1}{\sum\limits_{j=1}^{m}{{{l}_{j}}\left( 0 \right)}}\left( 1-{{e}^{-\frac{\alpha_{\rm th}}{{\delta }_{j}} }} \right),
\end{equation}
where ${{l}_{j}}\left( 0 \right)=\prod\limits_{k=1,k\ne j}^{m}{\frac{\delta_{j}}{\delta_{j}-\delta_{k}}}$ and
%where ${{l}_{j}}\left( 0 \right)=\prod\limits_{k=1,k\ne j}^{m}{\frac{\frac{1}{\delta_{k}}}{\frac{1}{\delta_{k}}-\frac{1}{\delta_{j}}}}$ and
\setlength\abovedisplayskip{1pt}
\setlength\belowdisplayskip{1pt}
\begin{equation}\label{eq30}
	{{\delta }_{j}}\approx  \frac{1-{{r}^{2}}}{1+{{r}^{2}}+2r\cos \left( \frac{\left( M+r \right)(M+1-j)\pi }{M\left( M+1 \right)} \right)} 
\end{equation}
for $j=1,\ldots ,M$.
\end{corollary}
\begin{proof}
For the exponential covariance matrix ${{\mathbf{R}}_{{{\mathbf{h}}}}}$, its eigenvalues ${{\delta }_{j}}$ for $M\gg 1$ can be approximated as Eq. \eqref{eq30} by following \cite[Eq. (51)]{mallik2018exponential}. 
According to the monotonicity of $\cos(x)$ for $0<x<\pi$, we have $\delta_1>\delta_2>\cdots>\delta_M>0$, and thus $B=r_M = M$, ${\bar T}_b=b$, $\bar{\delta}_{b,t}={\delta}_{t}$, ${\bar{r}_{b,t}}=1$ for $t=1,...,\bar{T}_b$. Via substituting these into Eq. \eqref{eqbeam} in Corollary \ref{Ltmodifed_beamdomain}, Eq. \eqref{eq29} can be obtained.
\end{proof}

%% 分析完改进方法的训练长度之后，说一下相比原来设计的优势。
The results in Eq. \eqref{eqtl},  Eq. \eqref{eqbeam} and  Eq. \eqref{eq29} are in closed-form and can be used to evaluate the average training length for different system parameter values.

In the following, we discuss the impact of the antenna number $M$ on the average training length of the modified beam-domain interleaved training.
\begin{corollary}\label{be_Lt_M}
For the modified beam-domain interleaved training scheme,  
when the use of all beams can avoid an outage, the average training length $L_t$ is a non-increasing function of the antenna number $M$, in both one-ring correlated channels with non-zero AS and exponentially correlated channels.
\end{corollary}
\begin{proof}
Denote the eigenvalues of the channel covariance matrix $\mathbf{R}_{\mathbf{h}}$ in descending order as $\lambda_{1},\lambda_{2},...,\lambda_{r_M}$ and $\Lambda_{1}, \Lambda_{2}, ..., \Lambda_{r_{M+1}}$ for the number of antennas being $M$ and $M+1$, respectively.
Since the channel covariance matrix for $M$ BS antennas is a submatrix of that for $M+1$ BS antennas, we have either $r_{M+1}=r_M+1$, e.g., for channels with the full-rank exponential covariance matrix in Eq. \eqref{eq:}, or $r_{M+1}=r_M$. 
From Eq. \eqref{eq11} and Eq. \eqref{eq15}, the average training length can be expressed as $
L_t = 1+\sum_{b=1}^{r_M-1}{\Pr \left( { \sum_{j=1}^{b}{{{\left| {{h}_{{\rm iid},j}} \right|}^{2}}{{\delta }_{j}}}}<\alpha_{\rm th}  \right)}$.
Then the difference between the average training lengths for systems with $M$ and $M+1$ BS antennas for the case of $r_{M+1}=r_M+1$ is    
\setlength\abovedisplayskip{1pt}
\setlength\belowdisplayskip{1pt}
\begin{equation}\label{delta_Lt}
\hspace{-0.2cm}L_t(M+1)-L_t(M) = \Delta + \Pr \left( { \sum\limits_{j=1}^{r_{M}}{{{\left| {{h}_{{\rm iid},j}} \right|}^{2}}{{\Lambda }_{j}}}}<\alpha_{\rm th}  \right), 
\end{equation}
where $\Delta=\sum_{b=1}^{r_M-1}\left[\Pr \left( { \sum_{j=1}^{b}{{{\left| {{h}_{{\rm iid},j}} \right|}^{2}}{{\Lambda }_{j}}}}<\alpha_{\rm th}  \right)-\Pr \left(  \right.\right.$ $\left.\left.\sum_{j=1}^{b}{{{\left| {{h}_{{\rm iid},j}} \right|}^{2}}{{\lambda }_{j}}}<\alpha_{\rm th}  \right)\right]$.
From the Eigenvalue Interlacing Theorem \cite{hwang2004cauchy}, we have $\Lambda_{r_{M+1}}\le\lambda_{r_M}\le\Lambda_{r_M}\le\lambda_{r_M-1}\le\Lambda_{r_M-1}\le\dots\le\lambda_{2}\le\Lambda_{2}\le\lambda_{1}\le\Lambda_{1}$. Thus $\Delta\le0$. The condition that the use of all beams can meet the transmission requirement leads to $\Pr \left( { \sum_{j=1}^{r_{M}}{{{\left| {{h}_{{\rm iid},j}} \right|}^{2}}{{\Lambda }_{j}}}}<\alpha_{\rm th}  \right)=0$.
Therefore, $L_t$ is non-increasing with increasing $M$ under this condition.
For the case of $r_{M+1} = r_M$, this conclusion still stands due to $L_t(M+1)-L_t(M) = \Delta$.
\end{proof}

\begin{remark}\label{be_Lt_M_re}
Numerical simulation based on Eq. \eqref{eqbeam} shows that the average training length $L_t$ increases with $M$ for small $M$; while for large $M$, $L_t$ decreases with $M$ and converges to a constant value. This is because when $M$ is small, $\Pr \left( { \sum_{j=1}^{r_{M}}{{{\left| {{h}_{{\rm iid},j}} \right|}^{2}}{{\Lambda }_{j}}}}<\alpha_{\rm th}  \right)$ is the dominant term in $L_t(M+1)-L_t(M)$ in Eq. \eqref{delta_Lt}, which has a positive value. When $M$ is large, $\Pr \left( { \sum_{j=1}^{r_{M}}{{{\left| {{h}_{{\rm iid},j}} \right|}^{2}}{{\Lambda }_{j}}}}<\alpha_{\rm th}  \right)\rightarrow 0$, therefore, as shown in Corollary \ref{be_Lt_M}, $L_t$ decreases with $M$. 	
\end{remark}

Next, we discuss the effect of the channel correlation on the average training length of the modified beam-domain interleaved training scheme for channels with the exponential covariance matrix.
Numerical simulation based on Eq. \eqref{eq29} shows that for relatively small $\alpha_{\rm th}$, higher channel correlation helps reduce the average training length of the modified beam-domain interleaved training. However, as $\alpha_{\rm th}$ continues to increase, an increase in the channel correlation may have the opposite effect.
According to the derivative of eigenvalues in Eq. \eqref{eq30}, i.e., ${{\delta }_{j}}=\frac{1-{{r}^{2}}}{1+{{r}^{2}}+2r\cos \left( \frac{\left( M+r \right)(M+1-j)\pi }{M\left( M+1 \right)} \right)}, \forall j=1, \ldots, M$,  with respect to $r$, larger eigenvalues increase, while smaller eigenvalues become smaller, as $r$ increases.
For very large $\alpha_{\rm th}$, $\Pr \left( { \sum_{j=1}^{b}{{{\left| {{h}_{{\rm iid},j}} \right|}^{2}}{{\delta }_{j}}}}<\alpha_{\rm th}  \right) \approx 1$ for $b< r_M-1$, and $\Pr \left( { \sum_{j=1}^{r_M-1}{{{\left| {{h}_{{\rm iid},j}} \right|}^{2}}{{\delta }_{j}}}}<\alpha_{\rm th}  \right)$  has the greatest impact on $L_t$. In this case, smaller $r$ results in flatter eigenvalue distribution which provides lower $\Pr \left( { \sum_{j=1}^{r_M-1}{{{\left| {{h}_{{\rm iid},j}} \right|}^{2}}{{\delta }_{j}}}}<\alpha_{\rm th}  \right)$ and shorter $L_t$. For small enough $\alpha_{\rm th}$, $\Pr \left( { \sum_{j=1}^{b}{{{\left| {{h}_{{\rm iid},j}} \right|}^{2}}{{\delta }_{j}}}}<\alpha_{\rm th}  \right) \approx 0$ for $2<b\le r_M-1$, and $\Pr \left( { {{{\left| {{h}_{{\rm iid},1}} \right|}^{2}}{{\delta }_{1}}}}<\alpha_{\rm th}  \right)$  has the greatest impact on $L_t$.  In this case, larger $r$ results in higher ${\delta }_{1}$ and shorter $L_t$.

\section{Modified Antenna-Domain Interleaved Training and Performance Analysis}\label{Modified Antenna-domain Interleaved Training}

In this section, we first discuss the impact of channel correlation on the average training length of the basic antenna-domain interleaved training. Then we derive the conditional distribution of channels of un-trained BS antennas based on channel values of the already trained BS antennas during the interleaved training process, based on which we further propose the design of the modified antenna-domain interleaved training.

\subsection{Average Training Length Analysis and Impact of Channel Correlation}\label{Impact of Channel Correlation}

In the following, we give closed-form expressions of the average training length of the basic antenna-domain interleaved training scheme under general correlated channels and exponentially correlated channels respectively.

Compared to Theorem \ref{Theorem1}, the only difference in the derivation on the average training length of the basic antenna-domain interleaved training is that the covariance matrix of the trained channels after $m$ training steps is $\tilde{\mathbf{R}}_m = \left[\mathbf{R}_{{{\mathbf{h}}}}^{\frac{1}{2}} \right]_{[1:m],[1:M]}$ and the vector of the trained channels can be represented as ${{\mathbf{h}}_{\mathbb{A}_{m}}}\cong \tilde{\mathbf{R}}_m\mathbf{h}_{\rm iid}$. 
Consider the compact eigenvalue decomposition: $\tilde{\mathbf{R}}_m^{{\rm H}}\tilde{\mathbf{R}}_m={{\mathbf{U}}_{m}} {{\mathbf{\Sigma }}_{m}} \mathbf{U}_{m}^{{\rm H}}$ where ${{\mathbf{\Sigma }}_{m}}= {\rm diag}\left\{ {{\delta }_{m,1}},\ldots ,{{\delta }_{m,r_m}}\right\}$ and $r_m$ is the rank of $\tilde{\mathbf{R}}_m^{{\rm H}}\tilde{\mathbf{R}}_m$. Then we have  ${{\left\|{{\mathbf{h}}_{\mathbb{A}_{m}}} \right\|}^{2}}\cong{{ \tilde{\mathbf{{\rm h}}}_{\rm iid}^{m,{\rm H}} }}{{\mathbf{\Sigma }}_{m}}\tilde{\mathbf{h}}_{\rm iid}^{m}$ with $\tilde{\mathbf{h}}_{\rm iid}^{m}\sim \mathcal{C}\mathcal{N}\left( \mathbf{0},{{\mathbf{I}}_{r_m}} \right)$.
Suppose that there are $T_m$ different eigenvalues with value of $\bar{\delta}_{m,t}$ and repeated times of $r_{m,t}$ for  $t=1,...,T_m$. Define $\mathbf{r}_m=[{{r}_{m,1}},...,{{r}_{m,T_m}}]^{\rm T}$. 
%% 下面给出一般相关信道下基本天线交替训练的平均长度表达式

\begin{theorem}\label{Theorem2}
 The average training length of the basic antenna-domain interleaved training scheme under general correlated channels can be expressed as 
\setlength\abovedisplayskip{1pt}
\setlength\belowdisplayskip{1pt}
\begin{equation}\label{eqLt_basic_ant}
\begin{aligned}
	\hspace{-0.3cm}L_t=1+\hspace{-0.1cm}&\sum\limits_{m=1}^{M-1}\hspace{-0.05cm}{\prod\limits_{t=1}^{T_m}\hspace{-0.1cm}\left(\frac{1}{{\bar{\delta }_{m,t}}}\right)^{{{r}_{m,t}}}}\hspace{-0.1cm}\sum\limits_{k=1}^{T_m}\hspace{-0.1cm}\sum\limits_{s=1}^{{{r}_{m,k}}}\hspace{-0.1cm}{{\left( -1 \right)}^{{{r}_{m,k}}-s}}{{\bar{\delta }_{m,k}}}^{{{r}_{m,k}}-s+1}\\
&\times{{\Psi }_{m,k,s,\mathbf{r}_m}}\left[ 1-{{e}^{-\frac{\alpha_{\rm th}}{{\bar{\delta }_{m,k}}}}}\sum\limits_{u=0}^{{{r}_{m,k}}-s}{\frac{{{\left(\frac{\alpha_{\rm th}}{{\bar{\delta }_{m,k}}}  \right)}^{u}}}{u!}} \right],
\end{aligned}
\end{equation}
where 
${{\Psi }_{m,k,s,\mathbf{r}_m}}={{\left( -1 \right)}^{{{r}_{m,k}}-1}}\sum\limits_{\mathbf{i}\in {{\Omega }_{m,k,s}}}\prod\limits_{n\ne k}\left( \begin{matrix}
{{i}_{n}}+{{r}_{m,n}}-1 \\
{{i}_{n}} \\
\end{matrix} \right)$ $\times{{\left( \frac{1}{{\bar{\delta }_{m,n}}}-\frac{1}{{\bar{\delta }_{m,k}}} \right)}^{-\left( {{i}_{n}}+{{r}_{m,n}} \right)}}$, $\mathbf{i}\hspace{-0.025cm}=\hspace{-0.025cm}{{\left[ {{i}_{1}},\ldots ,{{i}_{T_m}} \right]}^{\rm T}}$, and ${{\Omega }_{m,k,s}}=\left\{ \left[ {{i}_{1}},\ldots ,{{i}_{T_m}} \right]\in {{\mathbb{Z}}^{T_m}};\sum\limits_{j=1}^{T_m}{{{i}_{j}}=s-1,{{i}_{k}}=0,{{i}_{j}}\ge 0} \text{ for all } j \right\}.$
\end{theorem}
%\end{corollary}
\begin{proof}
Please refer to the proof of Theorem \ref{Theorem1}.	
\end{proof}

%%%%%%%%%%%%%%%%%%%%%%%%%%%%%%%%%%
%% 增加相关性对基本天线域IT训练的影响 %%
%%%%%%%%%%%%%%%%%%%%%%%%%%%%%%%%%%
To analyze the impact of channel correlation on the average training length, we consider two extreme cases, i.e., the i.i.d. channels with ${{\delta }_{m,i}}=1, \forall m=1,...,M, i=1,...,l$ and the fully correlated channels with ${{\delta }_{m,1}}=m$ and ${{\delta }_{m,i}}=0, \forall i=2,...,m$ for $m=1,...,M$.

For the i.i.d. channels, we have $r_m=m$, ${T}_m=1$, $\bar{\delta}_{m,1}=1$, ${{r}_{m,1}}=m$ and ${{\Psi }_{m,k,s,\mathbf{r}_m}}=\left(-1\right)^{m-1}, \text{ for }s=1$; and ${{\Psi }_{m,k,s,\mathbf{r}_m}}=0, \text{ for }s=2,\dots,m$. Via substituting these into Eq. \eqref{eqLt_basic_ant} in Theorem \ref{Theorem2}, we can obtain
\setlength\abovedisplayskip{1pt}
\setlength\belowdisplayskip{1pt}
\begin{equation}\label{Ltiid}
L_t^{{\rm (i.i.d.)}}=1+\sum\limits_{m=1}^{M-1}{\left(1-e^{-\alpha_{\rm th}}\sum\limits_{i=0}^{m-1}{\frac{\alpha_{\rm th}^i}{i!} }\right)}.
\end{equation}
According to the result in \cite[Theorem 2]{koyuncu2018interleaving}, we have $L_t^{{\rm (i.i.d.)}}\le 1+\alpha_{\rm th}$ for $M\rightarrow \infty$.

For the fully correlated channels, we have $r_m=1$, ${T}_m=1$, $\bar{\delta}_{m,1}=m$, ${{r}_{m,1}}=1$ and ${{\Psi }_{m,k,s,\mathbf{r}_m}}=1, \text{ for }s=1$; and ${{\Psi }_{m,k,s,\mathbf{r}_m}}=0, \text{ for }s=2,\dots,m$. Via substituting these into Eq. \eqref{eqLt_basic_ant} in Theorem \ref{Theorem2}, we can obtain
\setlength\abovedisplayskip{1pt}
\setlength\belowdisplayskip{1pt}
\begin{equation}\label{Ltfc}
L_t^{{\rm (FC)}}=1+\sum\limits_{m=1}^{M-1}{\left(1-e^{-\frac{\alpha_{\rm th}}{m}}\right)}.
\end{equation}
And the behavior of $L_t^{{\rm (FC)}}$ for $M\rightarrow \infty$ is given in the following corollary.
\begin{corollary}
	For $M\rightarrow \infty$, we have $1+\alpha_{\rm th}\gamma-\frac{\pi^2}{12}\alpha_{\rm th}^2\le L_t^{{\rm (FC)}}-\alpha_{\rm th}\ln M\le 1+\alpha_{\rm th}\gamma$, where $\gamma\approx0.5772$ is the Euler's constant.
\end{corollary}
\begin{proof}
Define $G(x)=x+e^{-x}$. For $m>0$, we have $1-e^{-\frac{\alpha_{\rm th}}{m}}-\frac{\alpha_{\rm th}}{m}=-\left(G\left(\frac{\alpha_{\rm th}}{m}\right)-G(0)\right)=-\int_{0}^{\frac{\alpha_{\rm th}}{m}}G'(u)du=-\int_{0}^{\frac{\alpha_{\rm th}}{m}}\left(1-e^{-u}\right)du\ge -\int_{0}^{\frac{\alpha_{\rm th}}{m}}udu=-\frac{\alpha_{\rm th}^2}{2m^2}$. Meanwhile, we have $1-e^{-\frac{\alpha_{\rm th}}{m}}-\frac{\alpha_{\rm th}}{m}\le 0$. Therefore, we can obtain $-\sum_{m=1}^{M-1}{\frac{\alpha_{\rm th}^2}{2m^2}}\le \sum_{m=1}^{M-1}{\left(1-e^{-\frac{\alpha_{\rm th}}{m}}\right)}-\sum_{m=1}^{M-1}{\frac{\alpha_{\rm th}}{m}}\le 0 $. Since $\sum_{m=1}^{\infty}{\frac{1}{m^2}}=\frac{\pi^2}{6}$ and $\gamma=\lim\limits_{M\to\infty}\left(\sum_{m=1}^{M}{\frac{1}{m}}-\ln M\right)\approx0.5772$, we have $1+\alpha_{\rm th}\gamma-\frac{\pi^2}{12}\alpha_{\rm th}^2\le L_t^{{\rm (FC)}}-\alpha_{\rm th}\ln M\le 1+\alpha_{\rm th}\gamma$ for $M\rightarrow \infty$.	
\end{proof}
\begin{remark}
When $M$ increases asymptotically, under independent channels, the average training length of the basic scheme $L_t^{{\rm (i.i.d.)}}$ has the upper bound $1+\alpha_{\rm th}$; while under fully correlated channels, the average training length of the basic scheme $L_t^{{\rm (FC)}}$ is proportional to $\ln M$, which implies the negative effect of channel correlation to the average training length of the basic training scheme.
Further, numerical calculations based on Eq. \eqref{Ltiid} and \eqref{Ltfc} show that when $\alpha_{\rm th}$ is small compared to $M$, the basic scheme has a shorter average training length in the i.i.d. channels. This mainly benefits from the higher antenna diversity gain. 
After $\alpha_{\rm th}$ reaches about the same size as $M$, 
the basic scheme has a shorter average training length in the fully correlated channels, where the consistency of antenna energy is more important than the diversity gain.
\end{remark}

\begin{corollary}\label{Lt_basic_antenna_appro}
For channels with the exponential covariance matrix in Eq. \eqref{eq:} and almost all values of $0< r<1$, the average training length of the basic antenna-domain interleaved training can be expressed as
\setlength\abovedisplayskip{1pt}
\setlength\belowdisplayskip{1pt}
\begin{equation}\label{eqo}
	L_t = 2-{{e}^{-\alpha_{\rm th} }}+\sum\limits_{m=2}^{M-1}{\sum\limits_{j=1}^{m}{{{l}_{m,j}}\left( 0 \right)}}\left( 1-{{e}^{-\frac{\alpha_{\rm th}}{{\delta }_{m,j}} }} \right),
\end{equation}
where ${{l}_{m,j}}\left( 0 \right)=\prod\limits_{k=1,k\ne j}^{m}{\frac{{\delta_{m,j}}}{{\delta_{m,j}}-{\delta_{m,k}}}}$ and
\setlength\abovedisplayskip{1pt}
\setlength\belowdisplayskip{1pt}
\begin{equation}\label{eqjm-1}
	{{{\delta }_{m,j}}}\approx
		\begin{cases}
		1-2r\cos \left( \frac{j\pi }{m+1} \right),  &\text{if } 0<r\ll 1
		\\  
		\frac{1-r}{2}{{\sec }^{2}}\left( \frac{j\pi }{2m} \right),  &\text{if } 0<1-r\ll 1 \text{ }\&\\
&\hspace{0.5em}{r\ne 1-\frac{6 m}{3 \sec ^2\left(\frac{j \pi}{2 m}\right)+2\left(m^2-1\right)}}\\		 
	 	\Phi(r,m,j), &\text{else} \text{ }\&\text{ } m-\sum_{i=1}^{m-1}\Phi(r,m,i) \\
&\hspace{3em}\ne \Phi(r,m,j) \\ 
	\end{cases}
\end{equation}
for $j=1,\ldots ,m-1$ and
\setlength\abovedisplayskip{1pt}
\setlength\belowdisplayskip{1pt}
\begin{equation}\label{eqmm}
	{{\delta }_{m,m}}\approx
	\begin{cases}
		1-2r\cos \left( \frac{m\pi }{m+1} \right),  &\text{if   } 0<r\ll 1\\  
		m-\frac{\left( {{m}^{2}}-1 \right)\left( 1-r \right)}{3},  &\text{if   } 0<1-r\ll 1\\
		m-\sum\limits_{i=1}^{m-1}\Phi(r,m,i), &\text{else}\\
	\end{cases}
\end{equation}
where 
%$\Phi(r,m,j)\triangleq m-\sum_{i=1}^{m-1} \frac{1-r^2 }{1+r^2+2 r^2 \cos \left(\frac{i \pi}{m}\right)+2 r(1-r) \cos \left(\frac{i \pi}{m+1}\right)}-\frac{1-r^2}{\left(1+r^2+2 r^2 \cos \left(\frac{j \pi}{m}\right)+2 r(1-r) \cos \left(\frac{j \pi}{m+1}\right)\right)}$.
$\Phi(r,m,i)\triangleq \frac{1-r^2}{1+r^2+2 r^2 \cos \left(\frac{i \pi}{m}\right)+2 r(1-r) \cos \left(\frac{i \pi}{m+1}\right)}$.
\end{corollary}
\begin{proof}
For the exponential covariance matrix, the approximations of ${{\delta }_{m,j}}, j=1,...,m$ can be written as Eq. \eqref{eqjm-1} and Eq. \eqref{eqmm} according to \cite[Eq. (35), Eq. (43a-b), Eq. (49a-b)]{mallik2018exponential}.
From the monotonicity of $\cos(x)$ in $0<x<\pi$ and that of $\sec^2(x)$ in  $0<x<\frac{\pi}{2}$, we have that for $0<r\ll 1$, $\delta_{m,j}, j=1,...,m$ are different from each other, while for $0<1-r\ll 1$, $\delta_{m,j}, j=1,...,m-1$ are different from each other, and $\delta_{m,m}$ is different from $\delta_{m,j}, j=1,...,m-1$ for $r\ne 1-\frac{6 m}{3 \sec ^2\left(\frac{j \pi}{2 m}\right)+2\left(m^2-1\right)}$. For intermediate $r$ values, $\delta_{m,j}, j=1,...,m-1$ are different from each other due to the monotonicity of $\cos(x)$ in $0<x<\pi$ as well, and $\delta_{m,m}$ is different from $\delta_{m,j}, j=1,...,m-1$ for $m-\sum_{i=1}^{m-1}\Phi(r,m,i) \ne \Phi(r,m,j)$.  
Therefore, we have $r_m= m$, ${T}_m=m$, $\bar{\delta}_{m,t}={\delta }_{m,t}$, ${{r}_{m,t}}=1$ for $t=1,...,{T}_m$. Via substituting these into Eq. \eqref{eqLt_basic_ant} in Theorem \ref{Theorem2}, Eq. \eqref{eqo} can be obtained.
\end{proof}

For almost all values of $0\hspace{-0.05cm}<\hspace{-0.1cm}r\hspace{-0.1cm}<\hspace{-0.05cm}1$, we can conduct faster evaluation and analysis for the average training length of the basic antenna-domain interleaved training scheme with Eq. \eqref{eqo} compared to that with Eq. \eqref{eqLt_basic_ant}. 
For large $r$ values satisfying ${r = 1-\frac{6 m}{3 \sec ^2\left(\frac{j \pi}{2 m}\right)+2\left(m^2-1\right)}},\forall j=1,...,m-1$ or intermediate $r$ values satisfying ${m-\sum_{i=1}^{m-1}\Phi(r,m,i) = \Phi(r,m,j)},\forall j\hspace{-0.2cm}=\hspace{-0.2cm}1,...,m-1$, 
 the training length can still be calculated according to Eq. \eqref{eqLt_basic_ant}.

\subsection{Derivations on the Conditional PDF of the Untrained Channels}

On the one hand, simulation based on Eq. \eqref{eqLt_basic_ant} shows that the channel correlation leads to an increase in the average training length of the basic antenna-domain interleaved training at the general rate requirement. On the other hand, if the channel correlation exists and the system knows it as a priori, we can use the correlation to improve the training efficiency. Specifically, the conditional PDF of the untrained channels can be derived for given values of the already trained channels.
Based on this conditional PDF, the choice of the BS antenna for the next training step can be optimized. In this subsection, we derive the conditional PDF.

\begin{lemma}\label{lemma1}
Given the channel values of the already trained BS antennas, ${h}_m, m\in\mathbb{A}=\left\{{a}_{1},{a}_{2},...,{a}_{|\mathbb{A}|}\right\}$, the conditional PDF of the un-trained channels ${h}_{n}|\mathbf{h}_{\mathbb{A}}, n\in \mathbb{M}-\mathbb{A}$ follows $\mathcal{CN} \left( \bar{\mu}_n,\bar{\sigma}_n^2\right)$ where
% 是否需要公式表述这个证明？
\setlength\abovedisplayskip{1pt}
\setlength\belowdisplayskip{1pt}
\begin{equation}\label{miu_all}
	\bar{\mu}_n=\left[ {{[\mathbf{R}_{\mathbf{h}}]}_{n,{{a}_{1}}}},{{[\mathbf{R}_{\mathbf{h}}]}_{n,{{a}_{2}}}},...,{{[\mathbf{R}_{\mathbf{h}}]}_{n,{{a}_{|\mathbb{A}|}}}} \right] {{\mathbf{R}}^{-1}_{{{\mathbf{h}}_{\mathbb{A}}}}}\mathbf{h}_{\mathbb{A}},
\end{equation}
\setlength\abovedisplayskip{1pt}
\setlength\belowdisplayskip{1pt}
\begin{equation}\label{sigma_all}
\begin{aligned}
	\bar{\sigma}_n^2=1-&\left[ {{[\mathbf{R}_{\mathbf{h}}]}_{n,{{a}_{1}}}},{{[\mathbf{R}_{\mathbf{h}}]}_{n,{{a}_{2}}}},...,{{[\mathbf{R}_{\mathbf{h}}]}_{n,{{a}_{|\mathbb{A}|}}}} \right] {{\mathbf{R}}^{-1}_{{{\mathbf{h}}_{\mathbb{A}}}}} \\
&\times{{\left[ {{[\mathbf{R}_{\mathbf{h}}]}_{{{a}_{1}},n}},{{[\mathbf{R}_{\mathbf{h}}]}_{{{a}_{2}},n}},...,{{[\mathbf{R}_{\mathbf{h}}]}_{{{a}_{|\mathbb{A}|}},n}} \right]}^{{\rm T}}}, 
\end{aligned}
\end{equation}
and 
$\left[{{\mathbf{R}}_{{{\mathbf{h}}_{\mathbb{A}}}}}\right]_{i,j}=[\mathbf{R}_{\mathbf{h}}]_{a_i,a_j}, i,j=1,\dots,\left|\mathbb{A}\right|$.
The conditional cumulative distribution function (CDF) of the power of the untrained channel $h_n,n\in \mathbb{M}-\mathbb{A}$ is
\setlength\abovedisplayskip{1pt}
\setlength\belowdisplayskip{1pt}
\begin{equation}\label{CDF}
	\Pr \left( {{\left| {{h}_{n}} \right|}^{2}}\le x |\mathbf{h}_{\mathbb{A}}\right)=1-{{Q}_{1}}\left( \sqrt{2}\frac{\left| {{{\bar{\mu }}}_{n}} \right|}{{{{\bar{\sigma }}}_{n}}},\sqrt{2}\frac{\sqrt{x}}{{{{\bar{\sigma }}}_{n}}} \right).
\end{equation}
\end{lemma}
\begin{proof}
${{\mathbf{R}}_{{{\mathbf{h}}_{\mathbb{A}}}}}$ is the covariance matrix of the vector of the trained channels $\mathbf{h}_{\mathbb{A}}$, which is a submatrix of the overall channel covariance matrix $\mathbf{R}_{\mathbf{h}}$.
Recall that $\mathbf{h}$ is a circular-symmetric complex Gaussian vector. Then from  \cite[Eq. (32)]{539051}, the conditional mean in Eq. \eqref{miu_all} and the conditional variance in Eq. \eqref{sigma_all} can be obtained. The CDF in Eq. \eqref{CDF} can be obtained from properties of noncentral chi-squared distribution, i.e., $\frac{\sqrt{2}}{{{{\bar{\sigma }}}_{n}}}{{h}_{n}}|\mathbf{h}_\mathbb{A}\sim\mathcal{C}\mathcal{N}\left( \sqrt{2}\frac{{{{\bar{\mu }}}_{n}}}{{{{\bar{\sigma }}}_{n}}},2 \right)$ due to  ${{h}_{n}}|\mathbf{h}_\mathbb{A}\sim \mathcal{C}\mathcal{N}\left( {{{\bar{\mu }}}_{n}},\bar{\sigma }_{n}^{2} \right)$. Then, we have the conditional PDF ${{\left| \frac{\sqrt{2}}{{{{\bar{\sigma }}}_{n}}}{{h}_{n}} \right|}^{2}}|\mathbf{h}_\mathbb{A}\sim{{\chi }^{2}}\left( 2,2\frac{{{\left| {{{\bar{\mu }}}_{n}} \right|}^{2}}}{{{{\bar{\sigma }}}^{2}}_{n}} \right)$ and the conditional CDF
\setlength\abovedisplayskip{1pt}
\setlength\belowdisplayskip{1pt}
\begin{equation}\label{pro}
	{{F}_{{{\left| \frac{\sqrt{2}}{{{{\bar{\sigma }}}_{n}}}{{h}_{n}} \right|}^{2}}|\mathbf{h}_\mathbb{A}}}\left( x \right)=1-{{Q}_{1}}\left( \sqrt{2}\frac{\left| {{{\bar{\mu }}}_{n}} \right|}{{{{\bar{\sigma }}}_{n}}},\sqrt{x} \right).
\end{equation}
Therefore, Eq. \eqref{CDF} can be obtained.
\end{proof}

Recall that $\mathbb{A} = \{a_1,...,a_{|\mathbb{A}|}\} $ denotes the set of indices of the already trained BS antennas and for simplicity of presentation, we assume that $a_1<a_2<\cdots<a_{|\mathbb{A}|}$.
If the index of an un-trained BS antenna $n$ satisfies $a_1<n<a_{|\mathbb{A}|}$, we denote the index of its nearest BS antennas in the trained set $\mathbb{A}$ with a smaller index as $a_{x^\star}$, that is, ${x^\star} = \arg \min_{a_i\in \mathbb{A}, a_i<n} n - a_i$. Thus $a_{x^\star+1}$ is the index of the trained BS antenna which is the nearest to Antenna $n$ with a larger index than $n$. Define $x_1 = n - a_{x^\star}$ and $x_2 = a_{{x^\star}+1} - n$.
\begin{corollary}\label{corollary3}
 Under the exponential correlation model, we have
\setlength\abovedisplayskip{1pt}
\setlength\belowdisplayskip{1pt}
\begin{equation}\label{miu_n}
\bar{\mu}_n=\begin{cases} {{\left(\rho^*\right)}^{{{a}_{1}}-n}}{{{h}}_{{{a}_{1}}}}, \hspace{3.8em}\text{ if } n < a_1
\\ \frac{\left[ {{\rho }^{x_1}}\left( 1-{{r }^{2x_2}} \right){{{{h}}}_{{{a}_{x^\star}}}}+{{\left(\rho^*\right) }^{x_2}}\left( 1-{{r}^{2x_1}} \right){{{{h}}}_{{{a}_{x^\star+1}}}} \right]}{1-{{r }^{2\left( x_1+x_2\right)}}}, \\
\hspace{9.4em}\text{ if } a_1 < n < a_{|\mathbb{A}|}
\\ {{\rho }^{{n-{a}_{|\mathbb{A}|}}}}{{{h}}_{{{a}_{|\mathbb{A}|}}}}, \hspace{4em}\text{ if } n > a_{|\mathbb{A}|} \\ 
\end{cases},
\end{equation}
and
\setlength\abovedisplayskip{1pt}
\setlength\belowdisplayskip{1pt}
\begin{equation}\label{sigma_n}
	\bar{\sigma}^2_n = \begin{cases} 1-{{r}^{2\left({{a}_{1}}-n \right)}}, &\text{ if } n < a_1
\\ \frac{\left(1-r^{2x_1}\right)\left(1-r^{2x_2}\right)}{1-{{r}^{ 2\left(x_1+x_2\right)}}}, &\text{ if } a_1 < n < a_{|\mathbb{A}|}
\\ 1-{{r}^{2 \left(n-{{a}_{|\mathbb{A}|}}\right) }}, &\text{ if } n > a_{|\mathbb{A}|} \\ 
\end{cases}.
\end{equation}
The conditional distribution of the channel of BS antenna $n\in \mathbb{M} - \mathbb{A}$ is only related to channel values of the two nearest BS antennas on both sides, $a_{{x^\star}}$ and $a_{{x^\star}+1}$.
\end{corollary}
\begin{proof}
See Appendix A.
\end{proof}

The results in Corollary \ref{corollary3} can help significantly reduce the computational complexity for the conditional CDF of the untrained channel power for scenarios with an exponential correlation model.

\subsection{Modified Antenna-Domain Interleaved Training Scheme}\label{modified antenna domain scheme}

Based on the conditional PDF of the untrained channels in Lemma \ref{lemma1}, we propose a modified antenna-domain interleaved training scheme where at the beginning of each training step, the BS antenna whose channel is to be trained is optimally selected. The basic idea is to use the channel values of the already trained antennas to calculate the probability of meeting the transmission requirement if each untrained BS antenna is selected. Then the antenna with the highest probability is chosen.

For the selection of the first antenna $n_0$ to be trained, we cannot use the same approach since no channel values have been obtained. Instead, we use the conditional variance in Eq. \eqref{sigma_all}.
Under the assumption that all antennas have the same average power, the first antenna to be trained can be the one resulting in the minimum overall conditional variance of other antennas, e.g., $n_0 = \arg \min_{m}\sum_{n=1,n\neq m}^{M}\bar{\sigma}^2_n$.
It can be seen from Eq. \eqref{sigma_n} that $n_0 = \left\lfloor \frac{M+1}{2} \right\rfloor$ under the exponential correlation model.

Recall that ${\mathbb{A}_m}$ and ${\mathbf{h}}_{\mathbb{A}_m}$ are the set of indices of the $m$ BS antennas whose channels have been trained and the obtained channel vector of these BS antennas after $m$ training steps.
At the beginning of the $(m+1)$-th training step, based on the already obtained channel vector $\mathbf{h}_{\mathbb{A}_m}$, the BS antenna to be trained for the $(m+1)$-th step is selected by the BS as follows. 
For each untrained BS antenna $n\in \mathbb{M} - \mathbb{A}_{m}$, the BS calculates the probability that the obtainment of the channel of BS antenna $n$ in the $(m+1)$-th training step can meet the transmission requirements as follows:
\setlength\abovedisplayskip{1pt}
\setlength\belowdisplayskip{1pt}
\begin{equation}\label{eqzc_47}
\begin{aligned}
\Pr &\left( {{\left| {{h}_{n}} \right|}^{2}}\ge {{\alpha }_{\rm th}}-{{\left\| {{\mathbf{h}}_{\mathbb{A}_{m}}} \right\|}^{2}} \right)\\
&={{Q}_{1}}\left( \sqrt{2}\frac{\left| {{{\bar{\mu }}}_{n}} \right|}{{{{\bar{\sigma }}}_{n}}},\sqrt{2}\frac{\sqrt{\alpha_{\rm th} -{{\left\| {{\mathbf{h}}_{\mathbb{A}_{m}}} \right\|}^{2}}}}{{{{\bar{\sigma }}}_{n}}} \right).
\end{aligned}
\end{equation}
Then, the BS selects the one with the highest probability among all untrained antennas, i.e., the index of the BS antenna for the $m+1$ training step is 
\setlength\abovedisplayskip{1pt}
\setlength\belowdisplayskip{1pt}
\begin{equation}\label{n_star}
n^{\star} = \arg \max_{n\in \mathbb{M}-\mathbb{A}_{m}}{\Pr}\left( {{\left| {{h}_{n}} \right|}^{2}}\ge {{\alpha }_{\rm th}}-{{\left\| {{\mathbf{h}}_{\mathbb{A}_{m}}} \right\|}^{2}} \right).
\end{equation}
The proposed modified antenna-domain interleaved training scheme is presented in Algorithm \ref{algorithm_2}, where the major difference to the basic scheme is in Step 4 on the antenna selection.
\begin{algorithm}[htp] 
%\setstretch{1.2}
\linespread{1} \selectfont
	\caption{Modified Antenna-Domain Interleaved Training Scheme} 
	\begin{algorithmic}[1]
		\item \textbf{Initialization:} $n^{\star} = n_0$; $\mathbb{A}_{1}=\left\{n^{\star} \right\} $;  $m=1$; BS sends a pilot for UE to acquire ${h}_{n^{\star}}$;
		\item \textbf{While} ${{\left\| {{\mathbf{h}}_{\mathbb{A}_{m}}} \right\|}^{2}}<\alpha_{\rm th} \And m<M$ \textbf{do} 
		\item \hspace{0.5cm}The UE sends one bit $0$ and $h_{n^{\star}}$ to the BS;
		\item \hspace{0.5cm}The BS calculates the probability value for each $n \in \mathbb{M}-\mathbb{A}_{m}$ according to Eq. \eqref{eqzc_47} and then decides the index of next training antenna $ n^{\star}$ according to Eq. \eqref{n_star}; 
		\item \hspace{0.5cm}The BS sends a pilot for the UE to acquire ${h}_{n^{\star}}$; 
		\item \hspace{0.5cm}$m=m+1$; $\mathbb{A}_{m}=\left\{ \mathbb{A}_{m-1},n^{\star}\right\} $; 
		\item \textbf{end}
		\item \textbf{if} $\left\Vert \mathbf{h}_{\mathbb{A}_{m}}\right\Vert  ^{2}\geq\alpha_{\rm th}$ 
		\item \hspace{0.5cm}The UE feeds back one bit $1$ $\&$ $h_{n^{\star}}$ to the BS;
		\item \hspace{0.5cm}The BS conducts downlink precoding according to Eq. \eqref{eq7}; 
		\item \textbf{else}
		\item \hspace{0.5cm}The UE feeds back one bit $0$ to the BS; 
		\item \textbf{end} 
	\end{algorithmic} \label{algorithm_2}  
\end{algorithm}

\subsubsection{Complexity Analysis}

The complexity of Algorithm \ref{algorithm_2} is mainly generated by Step $4$ in the loop from Step $2$ to Step $7$. Denote the training length for a random channel realization as $N$ satisfying $1\le  N\le M$. For the $(m+1)$-th training step, where $m< N$, the number of operations needed for calculating ${{\mathbf{R}}^{-1}_{{{\mathbf{h}}_{\mathbb{A}}}}}$ in Eq. \eqref{miu_all} and Eq. \eqref{sigma_all} scales as $m^3$, and the number of operations needed for the remaining matrix multiplications in Eq. \eqref{miu_all} and Eq. \eqref{sigma_all} scales as $(M-m)m^2$. Therefore, the complexity for calculating the conditional mean and variance for a training process with $N$ steps is $O(N^4+MN^3)$ and an upper bound on the complexity of Algorithm \ref{algorithm_2} is $O(M^4)$ since $N\le M$. For channels with the exponential covariance matrix, since the conditional mean and variance can be calculated according to Eq. \eqref{miu_n} and Eq. \eqref{sigma_n} without matrix inversion and matrix multiplication, an upper bound of the algorithm complexity is $O\left(M^2\right)$.

\subsubsection{Average Training Length}

Similar to the derivation of Eq. \eqref{eq11}, the average training length of the modified antenna-domain interleaved training scheme can be expressed as
\setlength\abovedisplayskip{1pt}
\setlength\belowdisplayskip{1pt}
\begin{equation}\label{eq49}
	L_t=1+\sum\limits_{m=1}^{M-1}{\Pr \left( {{\left\|{\mathbf{h}}_{\mathbb{A}_m}\right\|}^{2}}<\alpha_{\rm th}  \right)}.
\end{equation}
To derive the analytical or even closed-form expression of $L_t$ in Eq. \eqref{eq49},
the key is to calculate ${\Pr}({{\left\| {{\mathbf{h}}_{\mathbb{A}_m}} \right\|}^{2}}<\alpha_{\rm th})$.
From Step $4$ in Algorithm \ref{algorithm_2} and Eq. \eqref{eqzc_47} we know that one should first calculate the conditional mean $\bar{\mu}_n$ and conditional variance $\bar{\sigma}^2_n$ for $n\in \mathbb{M} - {\mathbb{A}_m}$ based on both ${\mathbb{A}_{m}}$ and ${\mathbf{h}}_{\mathbb{A}_{m}}$ to decide the antenna index $n^{\star}$ for the $(m+1)$-th training step. This makes the derivation of the PDF of ${\left\|{\mathbf{h}}_{\mathbb{A}_{m+1}}\right\|}^{2}$ challenging because ${\mathbb{A}_{m}}$ changes for each channel realization. In addition, the ${{Q}_{1}}\left(a,b \right)$ function involves a two-fold infinite series summation, resulting in an implicit relationship between $n^{\star}$ and ${\mathbb{A}_{m}}$, ${\mathbf{h}}_{\mathbb{A}_{m}}$. These all make the derivation of an analytical expression of $L_t$ in Eq. \eqref{eq49} intractable.

To circumvent the above difficulties, we introduce the deep neural network (DNN) $L_t=f\left(M,\mathbf{R}_{\mathbf{h}},\alpha_{\rm th};\boldsymbol{\Theta}\right)$ with $\boldsymbol{\Theta}$ being the network parameter matrix to model the function of $L_t$ with respective to system parameters, e.g., the BS antenna number $M$, the channel covariance matrix $\mathbf{R}_{\mathbf{h}}$ and the normalized SNR threshold $\alpha_{\rm th}=\left({2^{R_{\rm th}}-1}\right)/{P}$. 
For channels with the exponential covariance matrix, one can use the correlation coefficient $\rho$ to replace the input parameter $\mathbf{R}_{\mathbf{h}}$.
{The latter simulation results show that the function $f$ can be well-fitted by a DNN model with fully connected hidden layers.}
This deep learning-based approximation of the average training length can provide a faster performance evaluation of the modified antenna-domain interleaved training scheme compared to the Monte Carlo simulation.

\section{Simulation and Discussion}
\label{Simulation}

In this section, simulation results are shown for the proposed modified beam-domain and antenna-domain interleaved training schemes and their comparison with existing baseline schemes. The exponential correlation model in Eq. \eqref{eq:} is considered in Sections \ref{Beam-Domain Interleaved Training in Correlated Channels} to \ref{Comparison of Basic and Modified Antenna-Domain Interleaved Training}. The one-ring correlation model in Eq. \eqref{eq_onering} is considered in Section \ref{Basic and Modified Antenna/Beam-Domain Interleaved Training in One-ring Correlated Channels}.

\subsection{Beam-Domain Interleaved Training Under the Exponential Correlation Model}\label{Beam-Domain Interleaved Training in Correlated Channels}
Fig. \ref{fig:1} shows the average training lengths of the basic and modified beam-domain interleaved training schemes under the exponential correlation model, including the simulation values, the theoretical values in Eq. \eqref{eqtl} of Theorem \ref{Theorem1} and Eq. \eqref{eqbeam} of Corollary \ref{Ltmodifed_beamdomain}, and the approximate values in Eq. \eqref{eq29} of Corollary \ref{Ltmodifed_beam_appro}. We can see from Fig. \ref{fig:1_1} that the curves of simulation values and theoretical values match well for different scenarios. The curves of approximate values for $M=32, 64$ and $\rho=0.8$ in Fig. \ref{fig:1_2} have some gap with the simulation curves, while the gap for the case of $M=64$ is relatively small. This is because that Eq. \eqref{eq30} is a large-$M$ eigenvalue approximation. These observations verify our derivations in Section \ref{Average Training Length Analysis}.
\begin{figure}[htp]
	\centering
	\subfloat[Basic interleaved training scheme]{
		\includegraphics[scale=0.6]{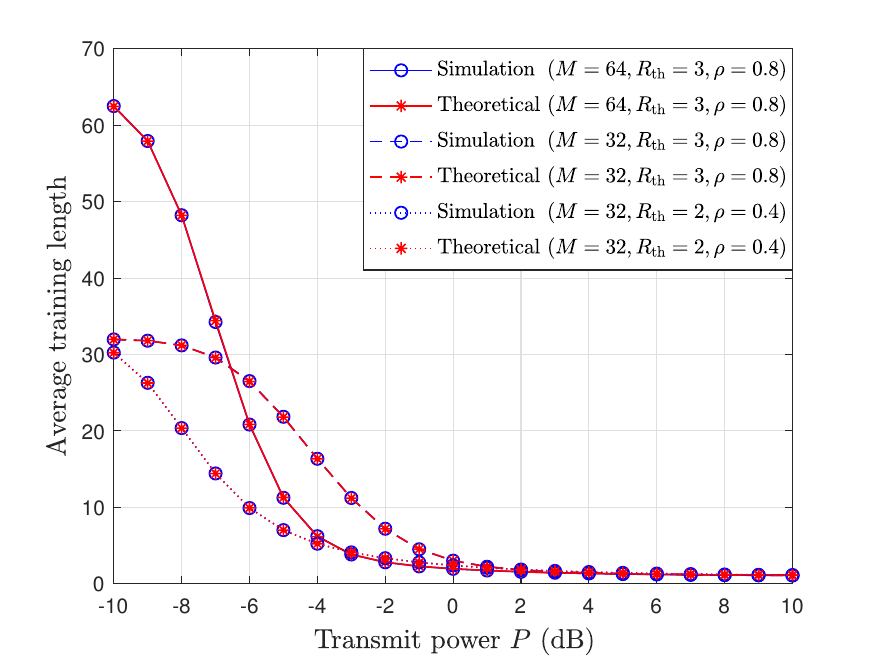}
		\label{fig:1_1}}
		
	\subfloat[Modified interleaved training scheme]{
		\includegraphics[scale=0.6]{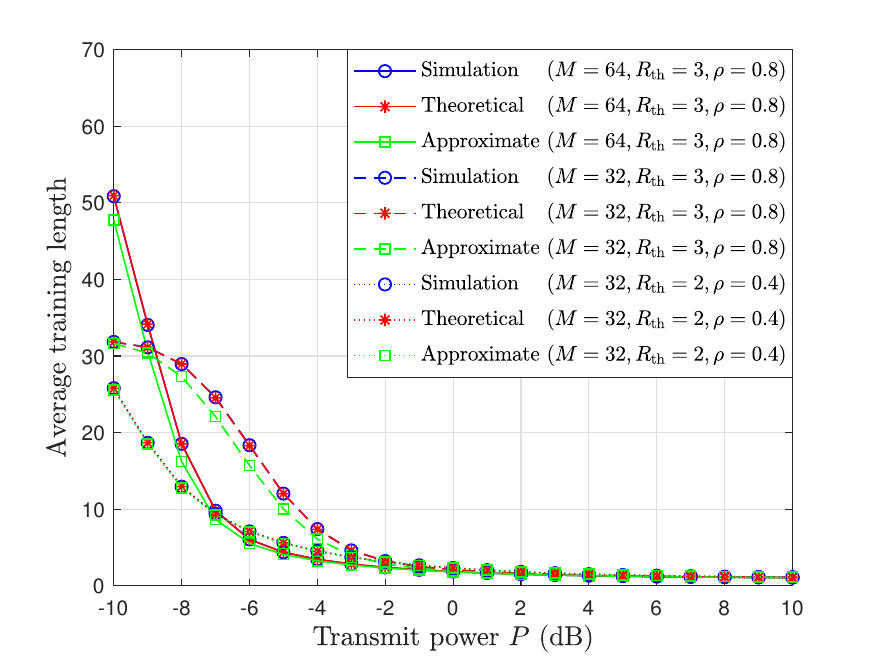}
		\label{fig:1_2}}
	\caption{Average training length of beam-domain interleaved training scheme.}
	\label{fig:1}	
\end{figure}

\subsection{Comparison of Basic and Modified Beam-Domain Interleaved Training Under the Exponential Correlation Model}\label{Performance Comparison of Basic and Modified Beam-Domain Interleaved Training}
\begin{figure}[htp]
	\centering
	\subfloat[Correlation coefficient $\rho=0.8$]{
		\includegraphics[scale=0.6]{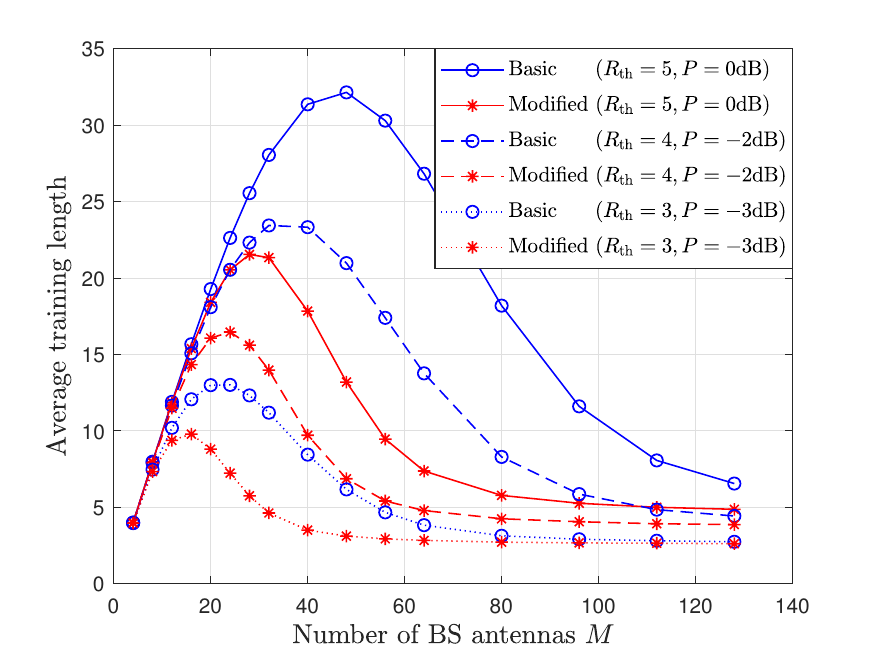}
		\label{fig:2_1}}

	\subfloat[Correlation coefficient $\rho=0.4$]{
		\includegraphics[scale=0.6]{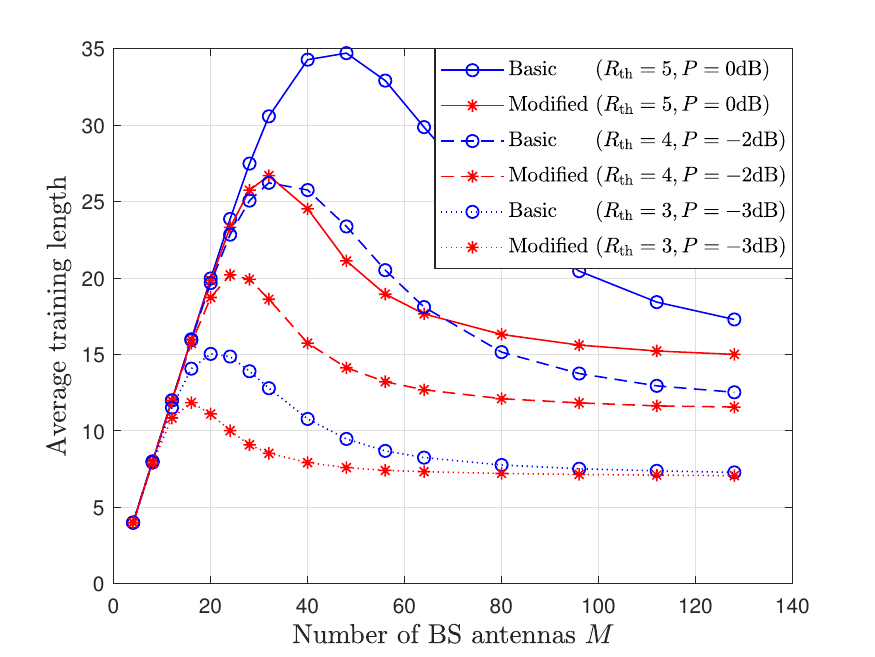}
		\label{fig:2_2}}
	\caption{Average training length of beam-domain interleaved training with different antennas number $M$.}
	\label{fig:2}
\end{figure}
Fig. \ref{fig:2} shows the average training lengths of the basic and modified beam-domain interleaved training with different antenna number $M$ for $\rho=0.8$ and $0.4$. It can be seen that the modified scheme outperforms the basic scheme in the average training length for three combinations of $R_{\rm th}$ and $P$, i.e., 1) $R_{\rm th} = 5$ bit/s/Hz, $P = 0$ dB and $\alpha_{\rm th} = 31$; 2) $R_{\rm th} = 4$ bit/s/Hz, $P = -2$ dB and $\alpha_{\rm th} = 23.77$; 3) $R_{\rm th} = 3$ bit/s/Hz,  $P = -3$ dB and $\alpha_{\rm th} = 13.97$. And the advantage becomes larger as $\alpha_{\rm th}$ increases from $\alpha_{\rm th} = 13.97$ to $\alpha_{\rm th} = 31$, showing that the modified beam-domain scheme exhibits greater performance advantages under more stringent transmission requirements. 
{In addition, the average training lengths of both schemes first increase and then decrease as $M$ increases. For large enough $M$, the training length levels off as $M$ increases. These results are consistent with the description in Corollary \ref{be_Lt_M} and Remark \ref{be_Lt_M_re}.}

\begin{figure}[htp]
\centering
\includegraphics[scale=0.6]{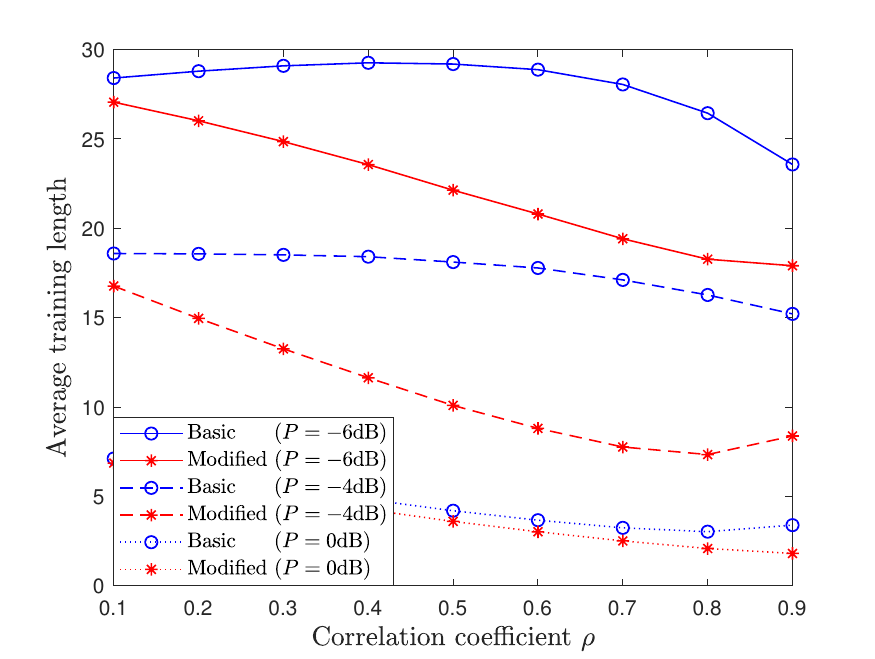}
\caption{Average training length of beam-domain interleaved training with different channel correlation levels.}
\label{fig:3}
\end{figure}
Fig. \ref{fig:3} shows the average training lengths of the basic and modified beam-domain interleaved training schemes with different channel correlation levels for $M=32$, $R_{\rm th}=3$ bit/s/Hz and $P = -6, -4, 0$ dB. Note that $\alpha_{\rm th} = 27.87, 17.58, 7$ respectively. The basic scheme uses the DFT codebook. The shorter average training length of the modified scheme can also be observed in the figure, especially for relatively low transmit power, e.g., $P=-4,-6$ dB. 
And with increasing $\rho$, the performance advantage of the modified scheme over the basic scheme enlarges for $P=0$ dB, while for $P=-4,-6$ dB, it first increases for $\rho\le 0.7$ and then decreases for $\rho > 0.7$.

\subsection{Antenna-Domain Interleaved Training Under the Exponential Correlation Model}\label{Antenna-Domain Interleaved Training in Correlated Channels}

\begin{figure}[htp]
	\centering
	\subfloat[Basic interleaved training scheme]{
		\includegraphics[scale=0.6]{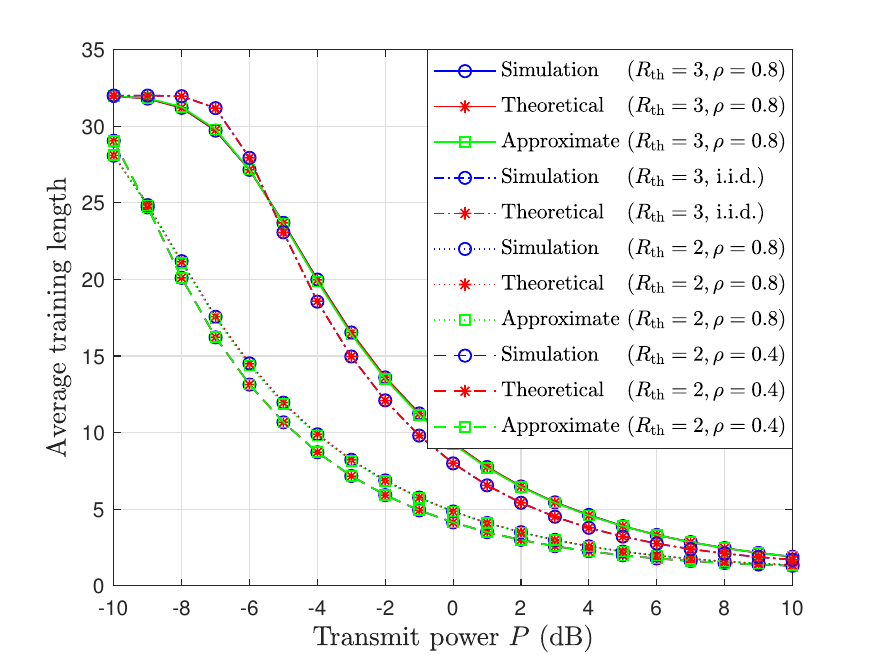}
		\label{fig:4_1}}

	\subfloat[Modified interleaved training scheme]{
		\includegraphics[scale=0.6]{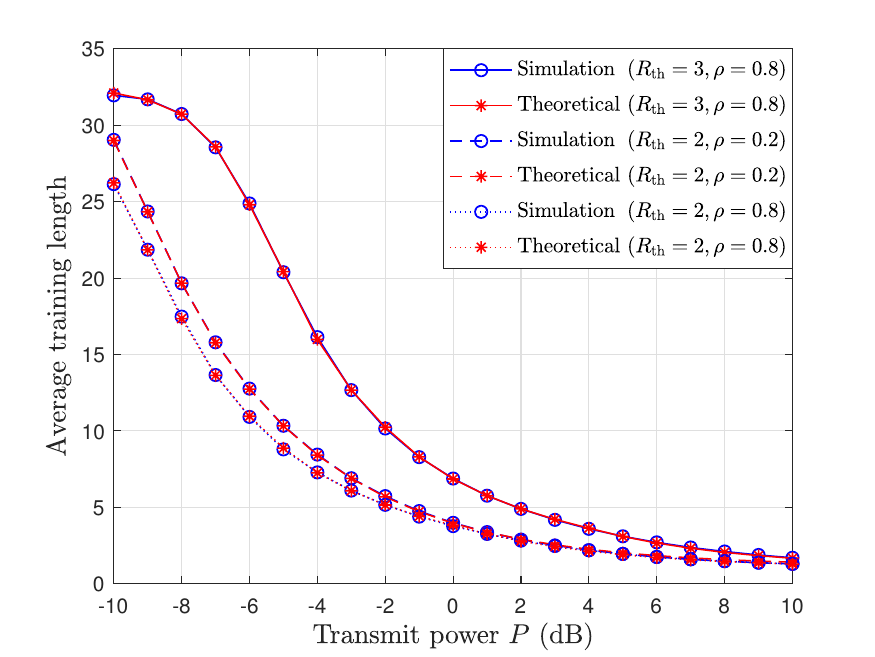}
		\label{fig:4_2}}
	\caption{Average training length of antenna-domain interleaved training scheme.}
	\label{fig:4}
\end{figure}
Fig. \ref{fig:4} shows the average training lengths of the basic and modified antenna-domain interleaved training schemes under the exponential correlation model for $M=32$, $R_{\rm th}=2, 3$ bit/s/Hz and $\rho = 0$, i.e., i.i.d. channels, $0.4$, and $0.8$.
Fig. \ref{fig:4_1} shows the simulation values, the theoretical values provided by Corollary  \ref{Theorem2}, i.e., Eq. \eqref{eqtl}, and the approximate values of the basic scheme in Eq. \eqref{eqo} of Corollary \ref{Lt_basic_antenna_appro}.
We can see from the figure that all three curves match well for different scenarios. This verifies the results in Section \ref {Impact of Channel Correlation}. 
For $R_{\rm th}=2$ bit/s/Hz, the average training length of the basic scheme increases when $\rho$ increases from $0.4$ to $0.8$ for $P>-9$ dB, i.e., $\alpha_{\rm th}<23.83$. For smaller $P$ and larger $\alpha_{\rm th}$, the increase of $\rho$, on the contrary, leads to a decrease of average training length for the basic scheme.

In Fig. \ref{fig:4_2}, we use a fully connected DNN containing four hidden layers (each with 4, 8, 16 and 32 Relu neurons) to provide an approximate average training length of the modified antenna-domain interleaved training scheme. 
The function for training the DNN model is the \textit{trainlm} based the Levenberg-Marquardt algorithm, which has the fastest convergence speed for medium-sized DNN. 
The loss is the mean-square error (MSE). The dataset has $3173$ samples of $<\rho,\alpha_{\rm th}, L_t >$. The ratio of the training set, validation set, and test set is 0.7:0.15:0.15. As shown in the figure, the designed DNN fits the function of training overhead well and its prediction of $L_t$ for these unseen combinations of $\rho$, $\alpha_{\rm th}$ and $P$ matches with the simulation results.

\subsection{Comparison of Basic and Modified Antenna-Domain Interleaved Training Under the Exponential Correlation Model}\label{Comparison of Basic and Modified Antenna-Domain Interleaved Training}

\begin{figure}[htp]
	\centering
	\subfloat[Correlation coefficient $\rho=0.8$]{
		\includegraphics[scale=0.6]{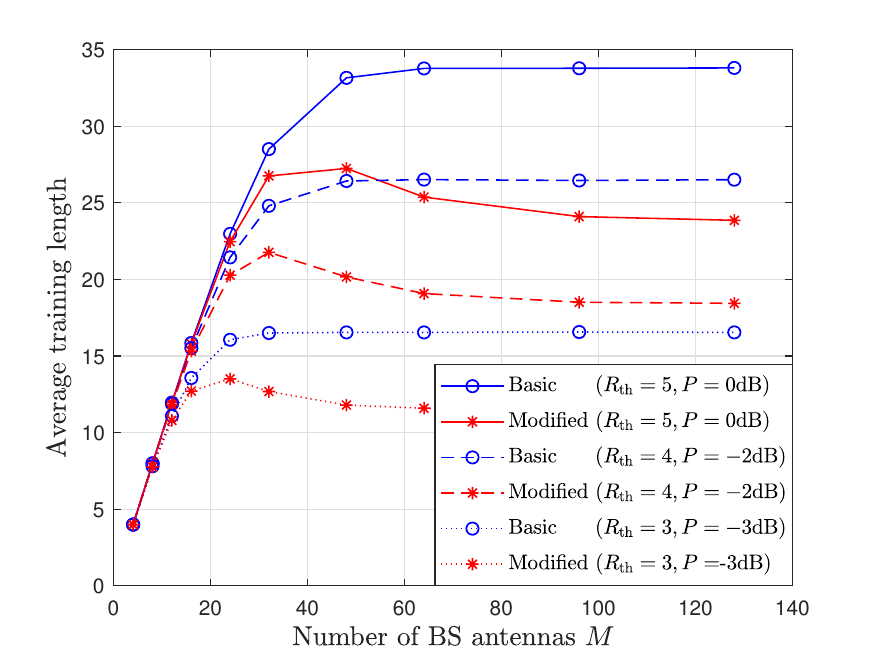}
		\label{fig:5_1}}

	\subfloat[Correlation coefficient $\rho=0.4$]{
		\includegraphics[scale=0.6]{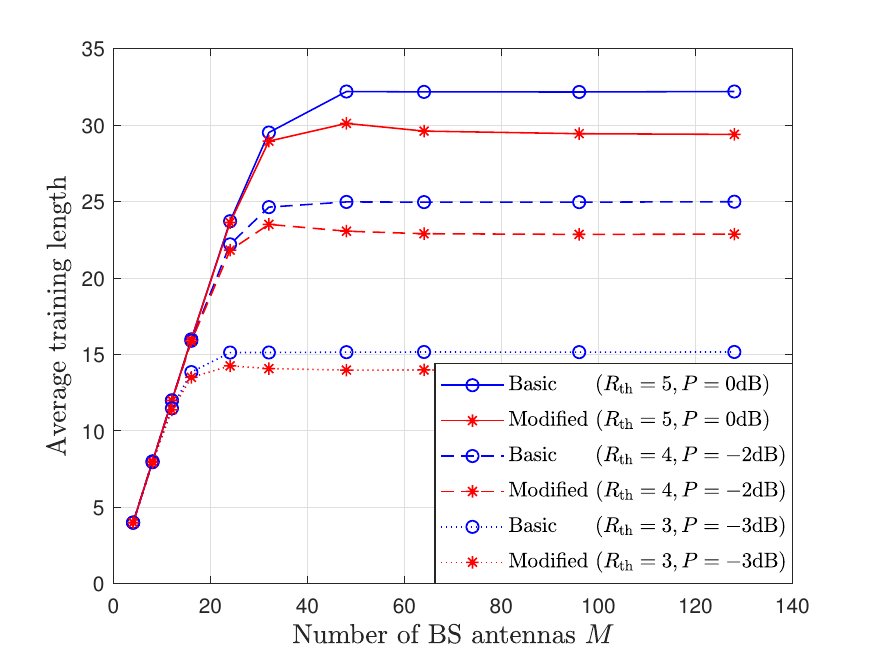}
		\label{fig:5_2}}
	\caption{Average training length of antenna-domain interleaved training with different BS antennas number $M$.}
	\label{fig:5}
\end{figure}
Fig. \ref{fig:5} shows the average training lengths of the basic and modified antenna-domain interleaved training schemes with different BS antenna number $M$ for $\rho=0.8$ and $0.4$.
It can be seen that the modified scheme outperforms the basic scheme in the average training length for three combinations of $R_{\rm th}$ and $P$, i.e., 1) $R_{\rm th} = 5$ bit/s/Hz, $P = 0$ dB and $\alpha_{\rm th} = 31$; 2) $R_{\rm th} = 4$ bit/s/Hz, $P = -2$ dB and $\alpha_{\rm th} = 23.77$; 3) $R_{\rm th} = 3$ bit/s/Hz,  $P = -3$ dB and $\alpha_{\rm th} = 13.97$. And the advantage becomes larger as $\alpha_{\rm th}$ increases from $\alpha_{\rm th} = 13.97$ to $\alpha_{\rm th} = 31$. 
In addition, with increasing $M$, the average training length of the basic scheme increases and converges, with faster convergence and a smaller value for $\rho = 0.4$ compared to those for $\rho =0.8$. However, the average training length of the modified scheme first increases and then decreases and finally levels off with increasing $M$.
{This is because as $M$ increases, there are more untrained antennas available after each interleaved training step, which increases the diversity of untrained antennas' conditional distributions.}
Furthermore, the performance advantage of the modified scheme over the basic scheme first increases and then levels off as $M$ increases.

\begin{figure}[htp]
\centering
\includegraphics[scale=0.6]{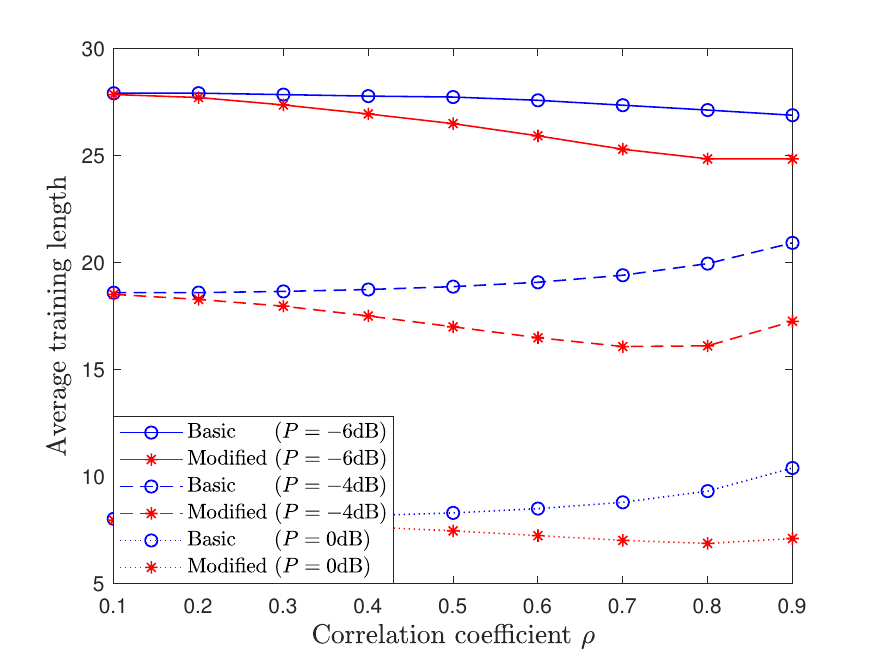}
\caption{Average training length of antenna-domain interleaved training with different channel correlation levels.}
\label{fig:6}
\end{figure}

\begin{figure*}[!b]
	\centering
	\subfloat[Channel AS $\sigma_{A} = 5^{\circ}$]{
		\includegraphics[scale=0.42]{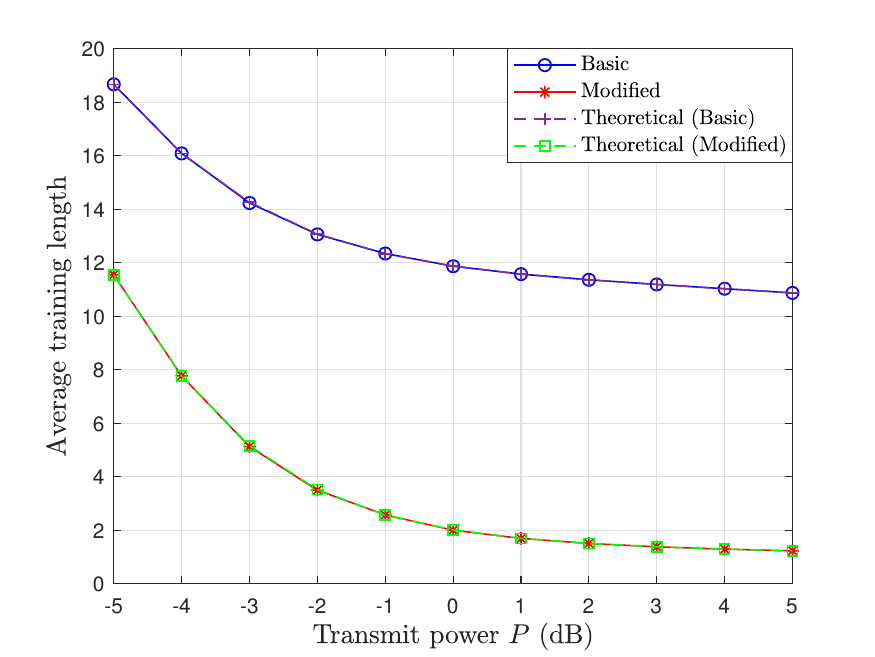}
		\label{fig:7_1}}
\hspace{-2.6em}	
	\subfloat[Channel AS $\sigma_{A} =10^{\circ}$]{
		\includegraphics[scale=0.42]{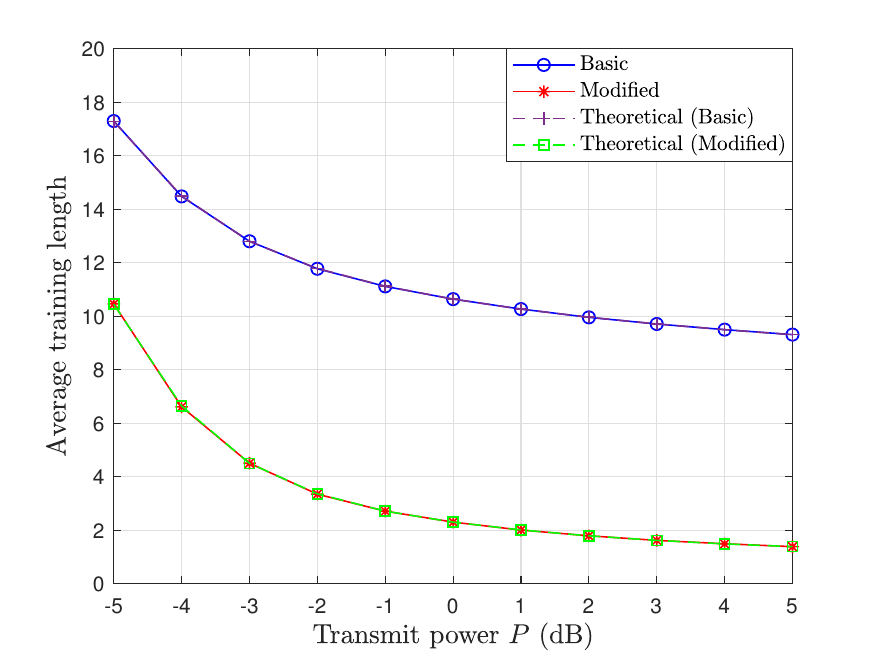}
		\label{fig:7_2}}
\hspace{-2.6em}	
	\subfloat[Channel AS $\sigma_{A} =20^{\circ}$]{
		\includegraphics[scale=0.42]{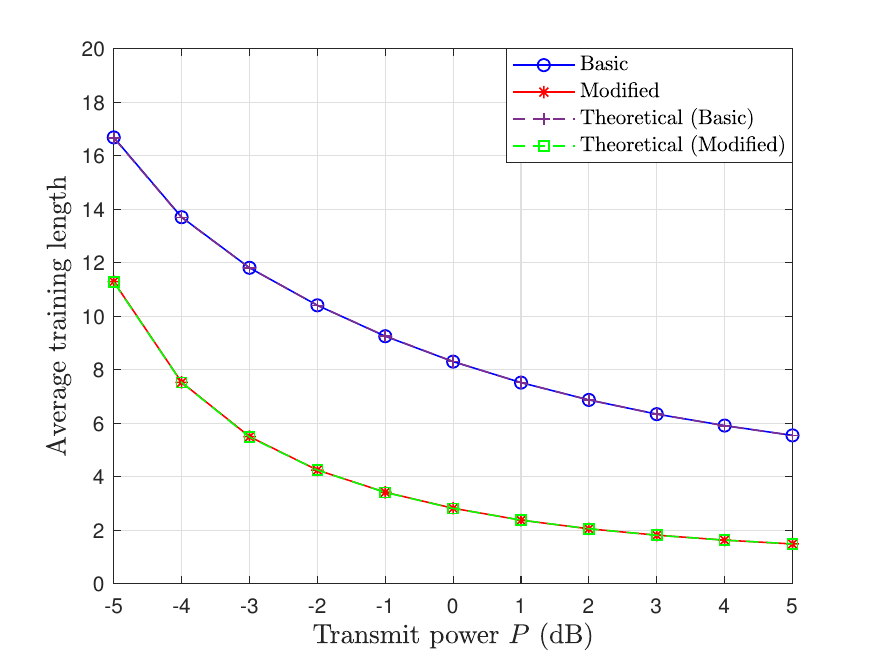}
		\label{fig:7_3}}
	\caption{Average training length of beam-domain interleaved training in one-ring correlated channels.}
	\label{fig:7}
\end{figure*}

Fig. \ref{fig:6} shows the average training lengths of the basic and modified antenna-domain interleaved training schemes with different channel correlation levels for $M=32$,  $R_{\rm th}=3$ bit/s/Hz and $P=-6, -4, 0$ dB.
We can see from the figure that the modified scheme has a shorter average training length than the basic scheme.
As the correlation coefficient $\rho$ increases, the average training length of the basic scheme increases for $P=0$ dB and $-4$ dB, but decreases for $P=-6$ dB. 
 On the contrary, with increasing $\rho$, the average training length of the modified scheme for $P=0, -4, -6$ dB first decreases for $\rho \le  0.8$ and then increases for larger $\rho > 0.8$.

 \begin{figure*}[!t]
	\centering
	\subfloat[Channel AS $\sigma_{A} =5^{\circ}$]{
		\includegraphics[scale=0.42]{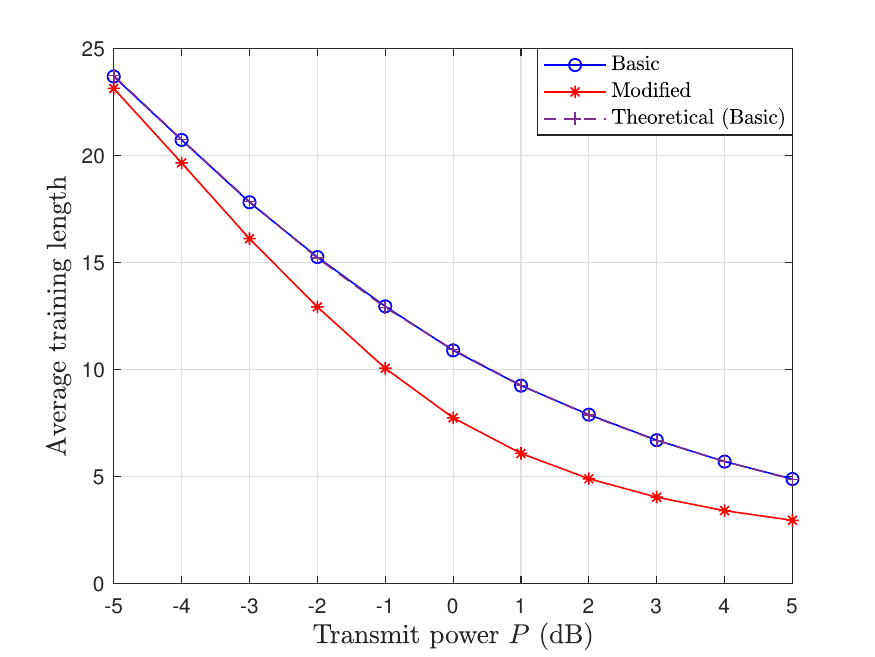}
		\label{fig:8_1}}
\hspace{-2.6em}	
	\subfloat[Channel AS $\sigma_{A} =10^{\circ}$]{
		\includegraphics[scale=0.42]{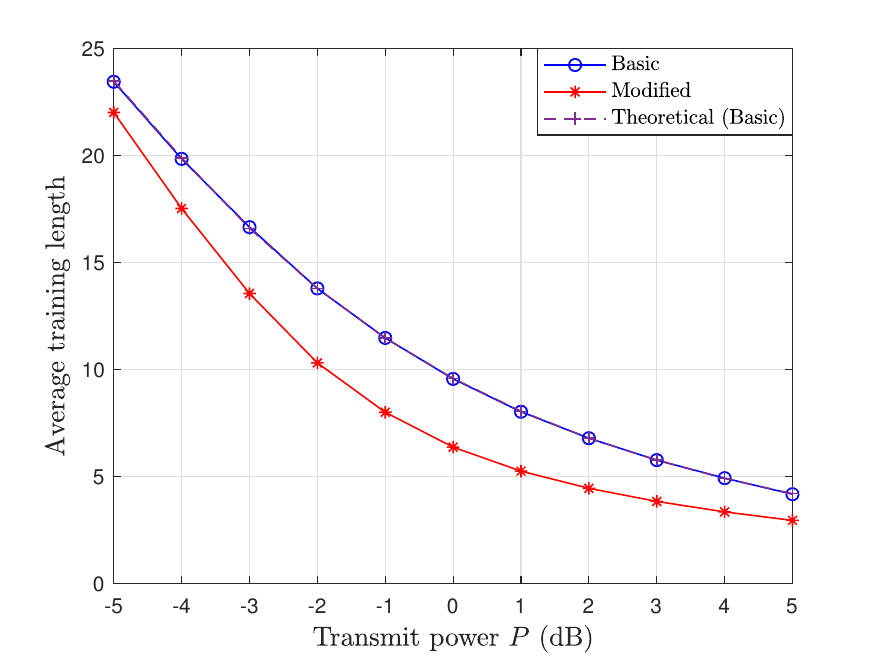}
		\label{fig:8_2}}
\hspace{-2.6em}	
	\subfloat[Channel AS $\sigma_{A} =20^{\circ}$]{
		\includegraphics[scale=0.42]{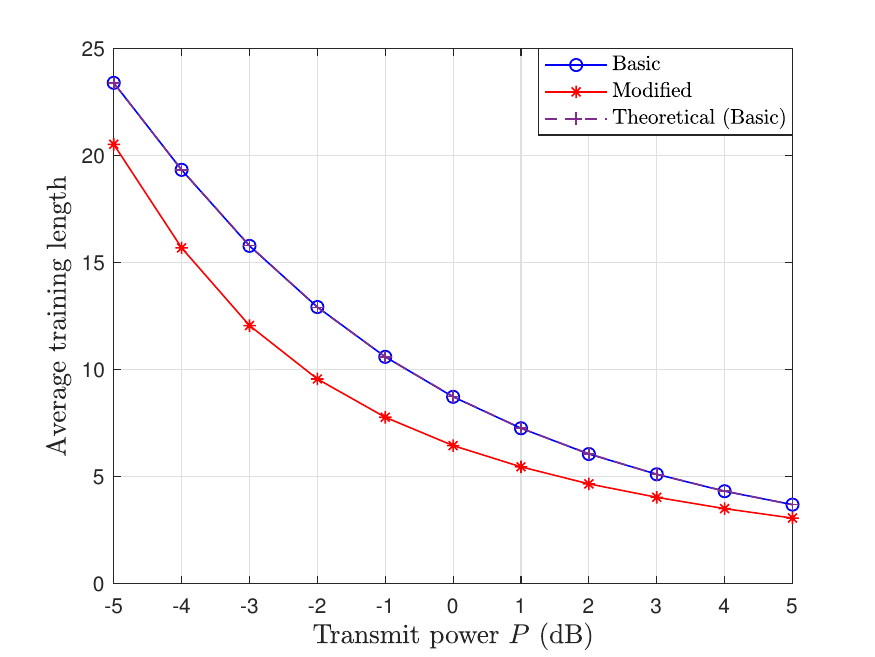}
		\label{fig:8_3}}
	\caption{Average training length of antenna-domain interleaved training scheme in one-ring correlated channels.}
	\label{fig:8}
\end{figure*}

\subsection{Antenna and Beam-Domain Interleaved Training Under the One-Ring Correlated Channels}\label{Basic and Modified Antenna/Beam-Domain Interleaved Training in One-ring Correlated Channels}

In this section, simulations are given to demonstrate the applicability of partial analytical results to the one-ring correlation model in Eq. \eqref{eq_onering}. The uniform PAS model \cite{schumacher2002antenna}, i.e., $f\left(\theta\right)=\frac{1}{2\Delta\theta}, \bar\theta-\Delta\theta\le \theta \le \bar\theta + \Delta\theta$, 
%Laplace distribution based PAS function, i.e., 
%\begin{equation}\label{eq_lap}
%f\left(\theta\right)=\frac{1}{\sqrt{2}\omega\left(1-e^{\left(-\frac{\sqrt{2}\pi}{\omega}\right)}\right)} e^{\left(\frac{-\sqrt{2}\left|\theta-\bar\theta\right|}{\omega}\right)},
%\end{equation}
is considered where $\bar\theta$ denotes the mean angular of departure (AoD) and the angular spread (AS) is $\sigma_{A} = \frac{\Delta\theta}{\sqrt{3}}$.

Fig. \ref{fig:7} shows the average training lengths of the basic and modified beam-domain interleaved training schemes with different transmit powers $P\in [-5,5]$ dB under the one-ring correlated channel  model for $M=32$, $R_{\rm th}=3$ bit/s/Hz, $D=0.5$, $\bar\theta=45^{\circ}$ and 
$\sigma_{A}  = 5^{\circ}, 10^{\circ}, 20^{\circ}$. 
The theoretical values of the average training length in Eq. \eqref{eqtl} of Theorem \ref{Theorem1} and Eq. \eqref{eqbeam} of Corollary \ref{Ltmodifed_beamdomain} match the simulation values well. And the modified scheme under three different ASs has an obvious performance advantage over the basic scheme. 
%In addition, 
%this performance advantage becomes greater as the AS decreases, i.e., the channel correlation increases, for $P\in [-5,5]$ dB.
With decreasing AS or increasing channel correlation, this performance advantage enlarges for relatively high transmit power, e.g., $P=5$ dB, while decreases for low transmit power, e.g., $P=-5$ dB.

Fig. \ref{fig:8} shows the average training lengths of the basic and modified antenna-domain interleaved training schemes with different transmit powers $P\in [-5,5]$ dB under the one-ring correlated channel model for $M=32$, $R_{\rm th}=3$ bit/s/Hz, $D=0.5$, $\bar\theta=45^{\circ}$ and 
$\sigma_{A}  = 5^{\circ}, 10^{\circ}, 20^{\circ}$.
The modified scheme in the antenna-domain under three different ASs also outperforms the basic scheme, and the theoretical average training length in Eq. \eqref{eqtl} matches the simulation value well.

\section{Conclusion}
\label{Conclusion}

In this paper, the channel spatial correlation was explored to improve the interleaved training for single-user massive MIMO downlink. 
Via optimizing the beam training codebook and the antenna training sequence based on the channel correlation, we respectively proposed the modified beam-domain and antenna-domain interleaved training schemes. 
For exponentially correlated channels, the conditional channel distribution of an untrained BS antenna given channel values of the already trained BS antennas was demonstrated to be only dependent on the channels of its nearest antennas on both sides in the already-trained antenna set, simplifying the complexity of the modiﬁed antenna-domain scheme significantly. 
Exact and approximate closed-form expressions were derived for the basic and modified beam/antenna-domain schemes in correlated channels. The impact of system parameters, e.g., the channel correlation, the antenna number, and the SNR requirement, on the average training length was explicitly revealed. 
Simulations verified our derivations and demonstrated the performance advantage of our proposed modified schemes.

{In addition to spatial correlation, channel temporal correlation can be exploited to improve the channel acquisition efficiency in massive MIMO systems. Unlike spatial correlation, temporal correlation has causality constraints in the time dimension and we can only use the historical training result to extrapolate. However, the historical result is not complete due to the characteristics of the interleaved scheme. How to explore the temporal correlation in interleaved training is worth to be further studied.
}

\appendix
\section{Appendix}
\subsection{Proof of Corollary \ref{corollary3}}\label{Appendix_c3}
Define
$\mathbf{m} \hspace{-0.05cm}=\hspace{-0.05cm} \left[\rho^{n-a_{1}},...,\rho^{n-a_{x^\star}},\left(\rho^*\right)^{a_{x^{\star}+1}-n},...,\left(\rho^*\right)^{a_{|\mathbb{A}|}-n}\right] $ $\times{{\mathbf{R}}^{-1}_{{{\mathbf{h}}_{\mathbb{A}}}}}=\left[m_1,\dots,m_{\left|\mathbb{A}\right|}\right].$
%\begin{equation}\label{m}
%	\mathbf{m} = \left[\rho^{n-a_{1}},...,\rho^{n-a_{x^\star}},\left(\rho^*\right)^{a_{x^{\star}+1}-n},...,\left(\rho^*\right)^{a_{|\mathbb{A}|}-n}\right] {{\mathbf{R}}^{-1}_{{{\mathbf{h}}_{\mathbb{A}}}}}=\left[m_1,\dots,m_{\left|\mathbb{A}\right|}\right].
%\end{equation}
We denote ${{\mathbf{R}}_{{{\mathbf{h}}_{\mathbb{A}}}}}$ as $\mathbf{R}$ for simplicity of presentation. Then we have
$
{{{m}}_{i}}=\frac{{\rm det}\left( {{\mathbf{R}}_{i}} \right)}{{\rm det}\left( \mathbf{R} \right)}, i\in\mathbb{I} = \{1,...,|\mathbb{A}|\},
$
where $\mathbf{R}_{i}$ is $\mathbf{R}$ by replacing its $i$-th row with $\left[\rho^{n-a_{1}},...,\rho^{n-a_{x^\star}},\left(\rho^*\right)^{a_{x^{\star}+1}-n},...,\left(\rho^*\right)^{a_{|\mathbb{A}|}-n}\right] $. Recall that $\mathbf{R}^{-1}=\frac{{\rm adj}\left( \mathbf{R} \right)}{\det\left(\mathbf{R}\right)}$ where ${\rm adj}\left( \mathbf{R} \right)$ is the adjugate matrix of $\mathbf{R}$, i.e., $\left[{\rm adj}\left( \mathbf{R} \right)\right]_{u,v}=R_{v,u}, \forall u,v=1,\dots,\left|\mathbb{A}\right|$ with $R_{u,v}$ being the algebraic cofactor of $[\mathbf{R}]_{u,v}$.
Therefore, we have ${{{m}}_{i}}=\frac{\sum\limits_{j=1}^{x^\star}{{{\rho }^{{n}-{{a}_{j}}}}}{{{R}}_{i,j}}+\sum\limits_{j=x^{\star}+1}^{|\mathbb{A}|}{{{\left(\rho^*\right)}^{{{a}_{j}}-{n} }}}{{{R}}_{i,j}}}{{\rm det}\left( \mathbf{R} \right)}=\frac{{\rm det}\left( {{\mathbf{R}}_{i}} \right)}{{\rm det}\left( \mathbf{R} \right)}$.

Here we prove that $m_{x^{\star}}$ and $m_{x^{\star}+1}$ are the only two non-zero elements in $\mathbf{m}$, equivalently ${\rm det}\left(\mathbf{R}_{i}\right)=0$ when $i\in\mathbb{I}-\{x^{\star},x^{\star}+1\}$ and ${\rm det}\left(\mathbf{R}_{i}\right)\ne 0$ when $i\in \{x^{\star},x^{\star}+1\}$.
For the first part, it is suffice to show that the row vectors of $\mathbf{R}_{i}$ are linearly dependent when $i\in\mathbb{I}-\{x^{\star},x^{\star}+1\}$ and we show this by construction.
Let $c_{x^{\star}}={{\rho }^{n-{{a}_{x^{\star}}}}}\frac{1-{{\left(\rho^*\rho\right) }^{{{a}_{x^{\star}+1}}-n }}}{1-{{\left(\rho^*\rho\right) }^{{{a}_{x^{\star}+1}}-{{a}_{x^{\star}}}}}},c_{x^{\star}+1}={{\left(\rho^*\right) }^{{{a}_{x^{\star}+1}}-n}}\frac{1-{{\left(\rho^*\rho\right) }^{n-{{a}_{x^{\star}}}}}}{1-{{\left(\rho^*\rho\right) }^{{{a}_{x^{\star}+1}}-{{a}_{x^{\star}}}}}} ,c_i\hspace{-0.05cm}=\hspace{-0.05cm}-1,c_j\hspace{-0.05cm}=\hspace{-0.05cm}0 \text{ for } j \notin \{x^\star,x^\star+1,i\}$, we have via straightforward calculations that $\sum_{j=1}^{|\mathbb{A}|}{c_j[\mathbf{R}_{i}]_{j,[1:|\mathbb{A}|]}}\hspace{-0.1cm}=\hspace{-0.1cm}c_{x^{\star}}[\mathbf{R}_{i}]_{x^{\star},[1:|\mathbb{A}|]}\hspace{-0.025cm}+\hspace{-0.025cm}c_{x^{\star}+1}[\mathbf{R}_{i}]_{x^{\star}+1,[1:|\mathbb{A}|]}-[\mathbf{R}_{i}]_{i,[1:|\mathbb{A}|]}=\mathbf{0}$.

Next we prove that ${{{m}}_{i}} \ne 0$ when $i\in \left\{ x^{\star},x^{\star}+1 \right\}$. Define $\mathbf{A}_{j}=\left[\mathbf{R}\right]_{[j:\left|\mathbb{A}\right|],\{1\}\cup[(1+j):\left|\mathbb{A}\right|]}$ for $j\in \left\{ 1,\ldots ,|\mathbb{A}|-1 \right\}$. Via splitting the $(1,2)$-th element in $\mathbf{A}_j$, i.e., the $(j,j+1)$-th element in $\mathbf{R}$, ${{\left(\rho^*\right)}^{{{a}_{j+1}}-{{a}_{j}}}}$ into $\left( {{\left(\rho^*\right) }^{{{a}_{j+1}}-{{a}_{j}}}}-{{\rho }^{{{a}_{j}}-{{a}_{j+1}}}} \right)+{{\rho }^{{{a}_{j}}-{{a}_{j+1}}}}$, we can split the ${\rm det}\left( \mathbf{A}_j \right)$ into the sum of two determinants and obtain the recurrence formula via expanding the first determinant in Eq. \eqref{EqdetAj} by the second column, i.e.,
\setlength\abovedisplayskip{1pt}
\setlength\belowdisplayskip{1pt}
\begin{equation}\label{EqdetAj} 
\begin{aligned}
{\rm det}&\left( \mathbf{A}_j \right)  = \\
&\left| \begin{matrix}
		{{\rho }^{{{a}_{j}}-{{a}_{1}}}} & {{\left(\rho^*\right) }^{{{a}_{j+1}}-{{a}_{j}}}}-{{\rho }^{{{a}_{j}}-{{a}_{j+1}}}} & \cdots  & {{\left(\rho^*\right) }^{{{a}_{|\mathbb{A}|}}-{{a}_{j}}}}  \\
		{{\rho }^{{{a}_{j+1}}-{{a}_{1}}}} & 0 & \cdots  & {{\left(\rho^*\right) }^{{{a}_{|\mathbb{A}|}}-{{a}_{j+1}}}}  \\
		\vdots  & \vdots  & \ddots  & \vdots   \\
		{{\rho }^{{{a}_{|\mathbb{A}|}}-{{a}_{1}}}} & 0 & \cdots  & 1  \\
	\end{matrix} \right| \\
  + &\left| \begin{matrix}
		{{\rho }^{{{a}_{j}}-{{a}_{1}}}}  & {{\rho }^{{{a}_{j}}-{{a}_{j+1}}}} & \cdots  & {{\left(\rho^*\right) }^{{{a}_{|\mathbb{A}|}}-{{a}_{j}}}}  \\
		{{\rho }^{{{a}_{j+1}}-{{a}_{1}}}} & 1 & \cdots  & {{\left(\rho^*\right) }^{{{a}_{|\mathbb{A}|}}-{{a}_{j+1}}}}  \\
		\vdots  & \vdots  & \ddots  & \vdots   \\
		{{\rho }^{{{a}_{|\mathbb{A}|}}-{{a}_{1}}}} & {{\rho }^{{{a}_{|\mathbb{A}|}}-{{a}_{j+1}}}} & \cdots  & 1  \\
	\end{matrix} \right| \\
=&-\left( {{\left(\rho^*\right) }^{{{a}_{j+1}}-{{a}_{j}}}}-{{\rho }^{{{a}_{j}}-{{a}_{j+1}}}} \right){\rm det}\left( \mathbf{A}_{j+1} \right).
\end{aligned}
\end{equation}
Note that the second determinant in Eq. \eqref{EqdetAj} is zero since the first column of its matrix is ${{\rho }^{{{a}_{j+1}}-{{a}_{1}}}}$ times the second column.
Then we calculate ${\rm det}\left( \mathbf{R} \right)={\rm det}\left( \mathbf{A}_1 \right)$ as follows:
\setlength\abovedisplayskip{1pt}
\setlength\belowdisplayskip{1pt}
\begin{equation} 
\begin{aligned}
{\rm det}\left( \mathbf{R} \right)&={\left( -1 \right)}^{\left|\mathbb{A}\right|-2}{\rm det}\left( \mathbf{A}_{\left|\mathbb{A}\right|-1} \right)\hspace{-0.1cm}\prod\limits_{j=1}^{\left|\mathbb{A}\right|-2}\hspace{-0.1cm}{{{\left(\rho^*\right) }^{{{a}_{j+1}}-{{a}_{j}}}}-{{\rho }^{{{a}_{j}}-{{a}_{j+1}}}}}\\
&=\prod\limits_{j=1}^{|\mathbb{A}|-1}{\left[ 1-{{\left(\rho^*\rho\right) }^{ {{a}_{j+1}}-{{a}_{j}} }} \right]}.
\end{aligned}
%\begin{aligned}
%{\rm det}\left( \mathbf{R} \right)&={\left( -1 \right)}^{\left|\mathbb{A}\right|-2}{\rm det}\left( \mathbf{A}_{\left|\mathbb{A}\right|-1} \right)\prod\limits_{j=1}^{\left|\mathbb{A}\right|-2}{{{\left(\rho^*\right) }^{{{a}_{j+1}}-{{a}_{j}}}}-{{\rho }^{{{a}_{j}}-{{a}_{j+1}}}}}\\
%&=\prod\limits_{j=1}^{|\mathbb{A}|-1}{\left[ 1-{{\left(\rho^*\rho\right) }^{ {{a}_{j+1}}-{{a}_{j}} }} \right]}.
%\end{aligned}
\end{equation}

Similar procedure can be used to calculate ${\rm det}\left( {{\mathbf{R}}_{x^{\star}}} \right)$. The difference is that the $(x^*,x^*+1)$-th element ${{\left(\rho^*\right) }^{{{a}_{{x^{\star}}+1}}-{n}}}$ in $ {{\mathbf{R}}_{x^{\star}}} $ is split into $\left( {{\left(\rho^*\right) }^{{{a}_{{x^{\star}}+1}}-{n}}}-{{\rho }^{{n}-{{a}_{{x^{\star}}+1}}}} \right)+{{\rho }^{{n}-{{a}_{{x^{\star}}+1}}}}$.
Then we can obtain
\setlength\abovedisplayskip{1pt}
\setlength\belowdisplayskip{1pt}
\begin{equation}
\begin{aligned}
{{{m}}_{x^{\star}}}&= \frac{{\rm det}\left( {{\mathbf{R}}_{x^{\star}}} \right)}{{\rm det}\left( \mathbf{R} \right)}=\frac{{{\left(\rho^*\right) }^{{{a}_{x^{\star}+1}}-n}}-{{\rho }^{n-{{a}_{x^{\star}+1}}}}}{{{\left(\rho^*\right) }^{{{a}_{x^{\star}+1}}-{{a}_{x^{\star}}}}}-{{\rho }^{{{a}_{x^{\star}}}-{{a}_{x^{\star}+1}}}}}\\
&={{\rho }^{n-{{a}_{x^{\star}}}}}\frac{1-{{\left(\rho^*\rho\right) }^{{{a}_{x^{\star}+1}}-n }}}{1-{{\left(\rho^*\rho\right) }^{{{a}_{x^{\star}+1}}-{{a}_{x^{\star}}}}}} \ne 0,
\end{aligned}
\end{equation}
for $\left|\rho\right|<1$ and $\left|\rho\right|\ne 0$.
$\det \left( \mathbf{R} \right)$ can also be calculated by splitting the $\left( j+1,j \right)$-th element ${{\rho }^{{{a}_{j+1}}-{{a}_{j}}}}$ into $\left( {{\rho }^{{{a}_{j+1}}-{{a}_{j}}}}-{{\left(\rho^*\right)}^{{{a}_{j}}-{{a}_{j+1}}}} \right)+{{\left(\rho^*\right) }^{{{a}_{j}}-{{a}_{j+1}}}}$. Then we split the determinant by the $\left( j+1 \right)$-th row into the sum of two determinants and leave the rest of the rows unchanged. Similarly, we calculate ${\rm det}\left( {{\mathbf{R}}_{x^{\star}+1}} \right)$ in this way. The difference is that the $\left( {x^{\star}+1},{x^{\star}} \right)$-th element ${{\rho }^{{{n}-{a}_{{x^{\star}}}}}}$ in $ {{\mathbf{R}}_{x^{\star}+1}}$ is split into $\left( {{\rho }^{{n}-{{a}_{{x^{\star}}}}}}-{{\left(\rho^*\right)}^{{{{a}_{{x^{\star}}}}}-n}} \right)+{{\left(\rho^*\right) }^{{{{a}_{{x^{\star}}}}}-n}}$. Then we obtain
\setlength\abovedisplayskip{1pt}
\setlength\belowdisplayskip{1pt}
\begin{equation} 
\begin{aligned}
{{{m}}_{x^{\star}+1}}& = \frac{{\rm det}\left( {{\mathbf{R}}_{x^{\star}+1}} \right)}{{\rm det}\left( \mathbf{R} \right)}=\frac{{{\rho }^{n-{{a}_{x^{\star}}}}}-{{\left(\rho^*\right) }^{{{a}_{x^{\star}}}-n}}}{{{\rho }^{{{a}_{x^{\star}+1}}-{{a}_{x}}}}-{{\left(\rho^*\right) }^{{{a}_{x^{\star}}}-{{a}_{x^{\star}+1}}}}}\\
&={{\left(\rho^*\right) }^{{{a}_{x^{\star}+1}}-n}}\frac{1-{{\left(\rho^*\rho\right) }^{n-{{a}_{x^{\star}}}}}}{1-{{\left(\rho^*\rho\right) }^{{{a}_{x^{\star}+1}}-{{a}_{x^{\star}}}}}} \ne 0.
\end{aligned}
\end{equation}

\ifCLASSOPTIONcaptionsoff
  \newpage
\fi

\bibliographystyle{IEEEtran}
%\bibliography{ReferencesNewAbbr}

\begin{IEEEbiography}[{\includegraphics[width=1in,height=1.25in,clip,keepaspectratio]{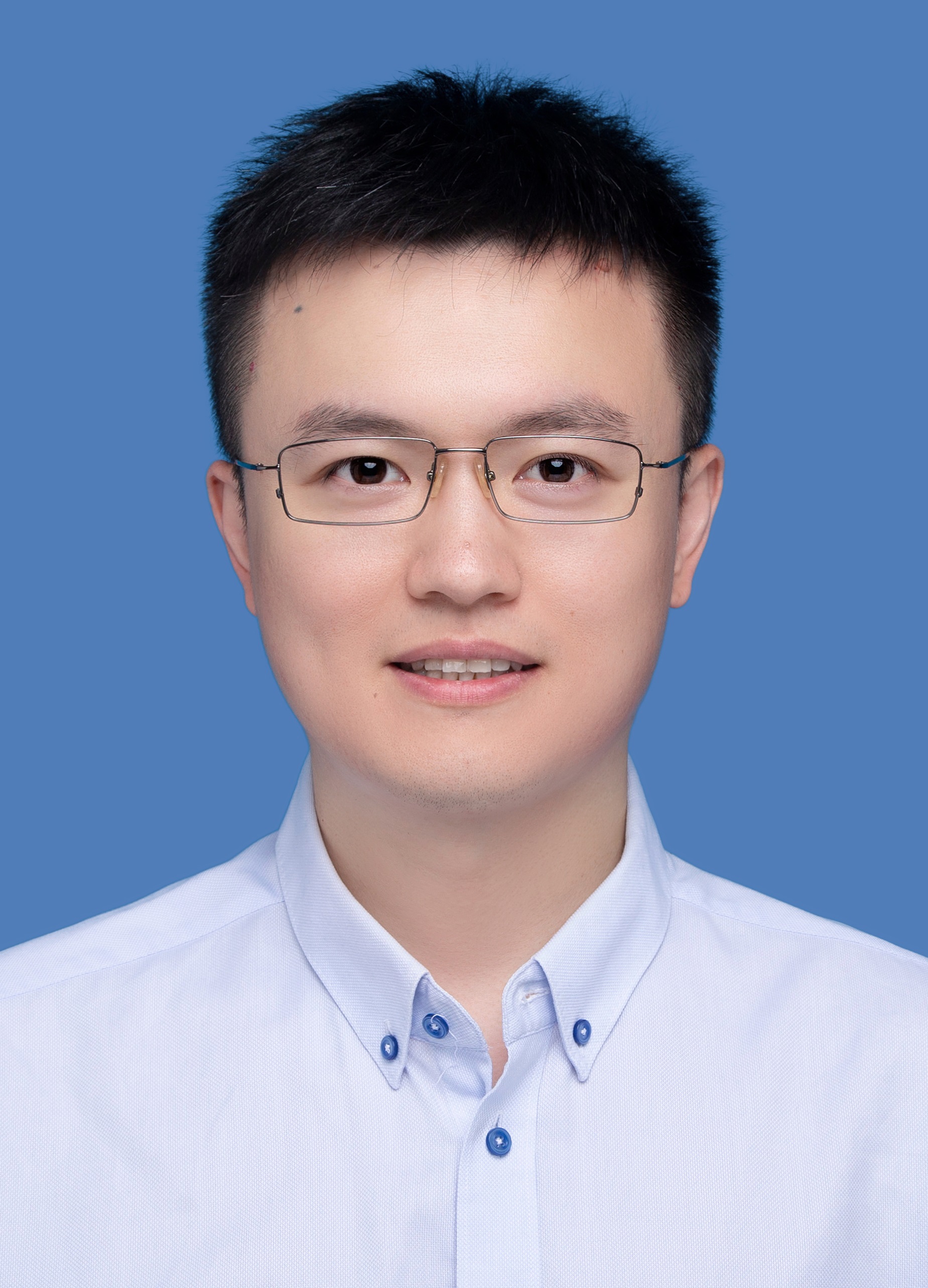}}]{Cheng Zhang}
	(Member, IEEE) received the B.Eng. degree from Sichuan University, Chengdu, China, in June 2009, the M.Sc. degree from the Xi’an Electronic Engineering Research Institute (EERI), Xi’an, China, in May 2012, and the Ph.D. degree from Southeast University (SEU), Nanjing, China, in Dec. 2018. From Nov. 2016 to Nov. 2017, he was a Visiting Student with the University of Alberta, Edmonton, AB, Canada.
	
	From June 2012 to Aug. 2013, he was a Radar Signal Processing Engineer with Xi’an EERI. Since Dec. 2018, he has been with SEU, where he is currently an Associate Professor, and supported by the Zhishan Young Scholar Program of SEU. His current research interests include space-time signal processing and machine learning-aided optimization for B5G/6G wireless communications. He has authored or co-authored more than 50 IEEE journal papers and conference papers. He was the recipient of the excellent Doctoral Dissertation of the China Education Society of Electronics in Dec. 2019, that of Jiangsu Province in Dec. 2020, and the Best Paper Awards of 2023 IEEE WCNC and 2023 IEEE WCSP.
\end{IEEEbiography}

\begin{IEEEbiography}[{\includegraphics[width=1in,height=1.25in,clip,keepaspectratio]{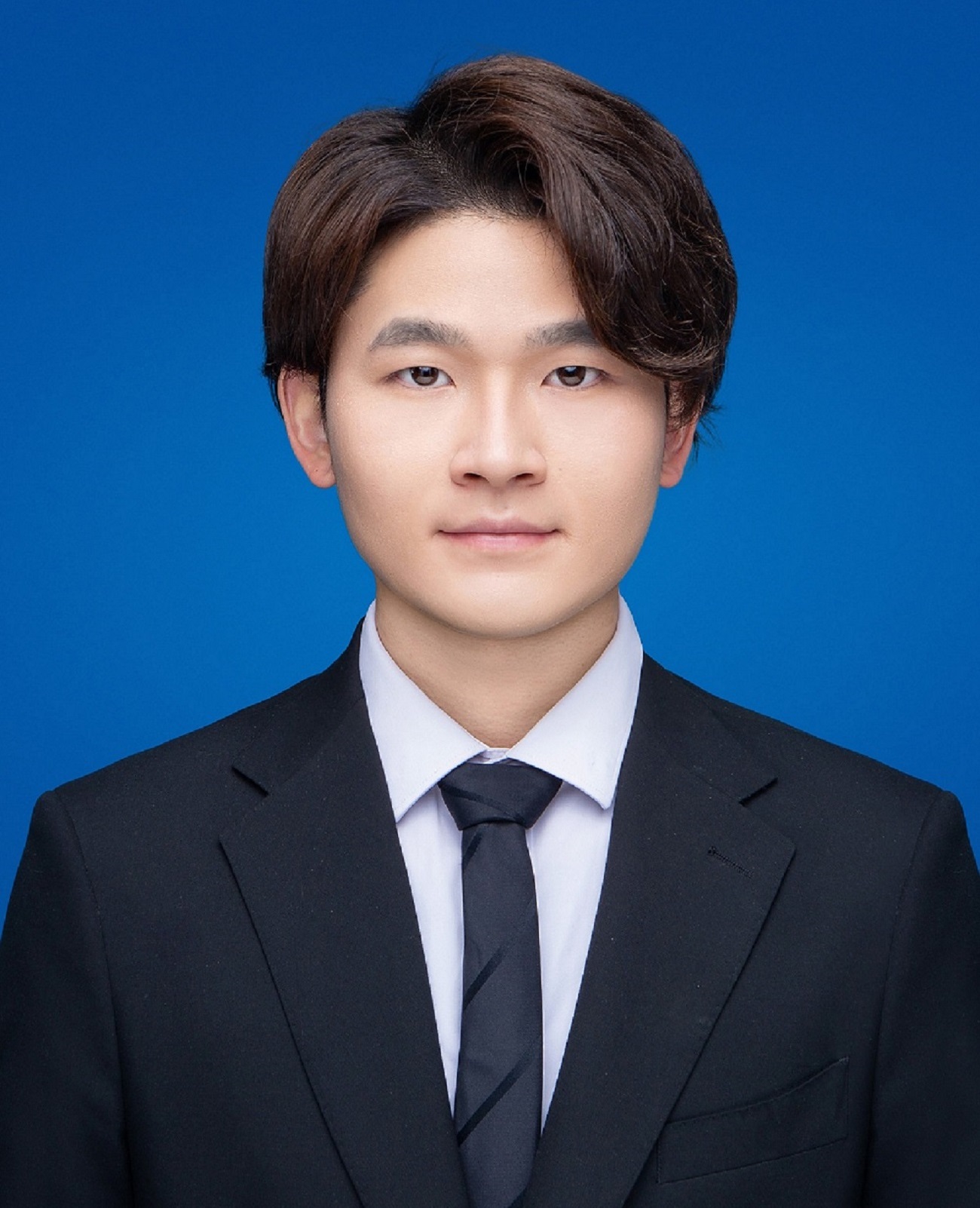}}]{Chang Liu}
	(Student Member, IEEE) received the B.Eng. degree in information engineering with the School of Information Science and Engineering, Southeast University, Nanjing, China, in 2021, where he is currently pursuing the M.Sc. degree as well. His research interests mainly focus on low overhead massive MIMO channel acquisition.
\end{IEEEbiography}

\begin{IEEEbiography}[{\includegraphics[width=1in,height=1.25in,clip,keepaspectratio]{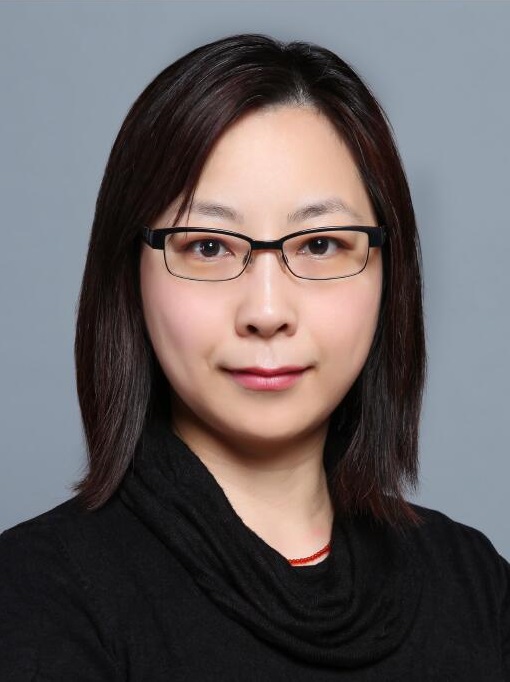}}]{Yindi Jing}
	(Senior Member, IEEE) received the B.Eng. and M.Eng. degrees in automatic control
	from the University of Science and Technology of China, Hefei, China, in 1996 and 1999, respectively. She received the M.Sc. degree and the Ph.D. in electrical engineering from California Institute of Technology, Pasadena, CA, in 2000 and 2004, respectively. From October 2004 to August 2005, she was a postdoctoral scholar at the Department of Electrical Engineering of California Institute of Technology. From February 2006 to June 2008, she was a postdoctoral scholar at the Department of Electrical Engineering and Computer Science of the University of California, Irvine. In 2008, she joined the Electrical and Computer Engineering Department of the University of Alberta, where she is currently a professor. She was an Associate Editor for IEEE TRANSACTIONS ON WIRELESS COMMUNICATIONS and a Senior Area Editor for IEEE SIGNAL PROCESSING LETTERS. She was a member of the IEEE Signal Processing Society Signal Processing for Communications and Networking (SPCOM) Technical Committee and a member of the NSERC Discovery Grant Evaluation Group for Electrical and Computer Engineering. Her research interests are in wireless communications and signal processing.
\end{IEEEbiography}

\begin{IEEEbiography}[{\includegraphics[width=1in,height=1.25in,clip,keepaspectratio]{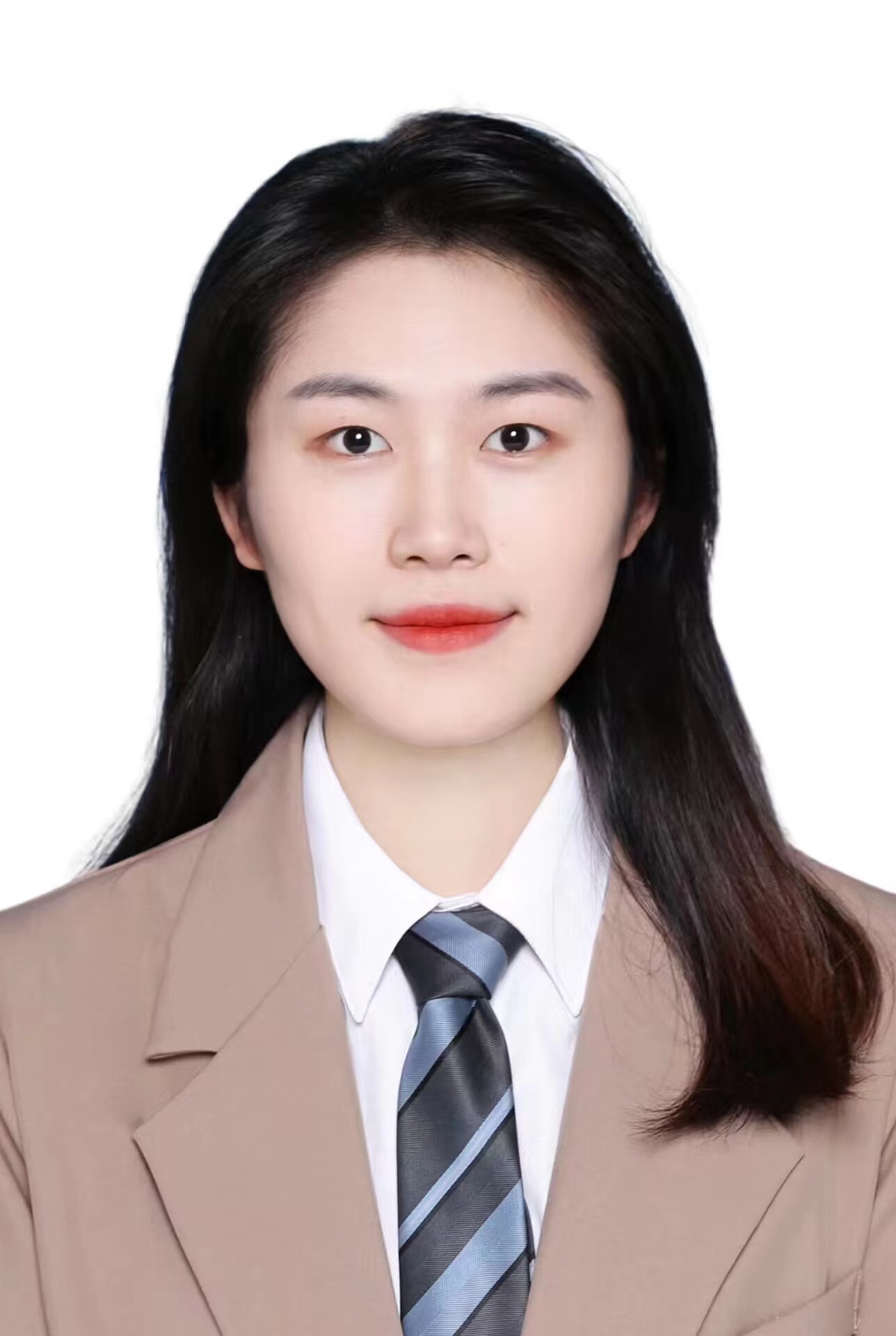}}]{Minjie Ding}
	(Student Member, IEEE) received the B.Eng. degree in information engineering from the School of Electronic Information and Engineering, Nanjing University of Aeronautics and Astronautics, Nanjing, China, in 2020, where she is currently pursuing the M.Sc. degree in information and communication engineering with the School of Information Science and Engineering, Southeast University. Her research interests mainly focus on low overhead massive MIMO channel acquisition.
\end{IEEEbiography}

\begin{IEEEbiography}[{\includegraphics[width=1in,height=1.25in,clip,keepaspectratio]{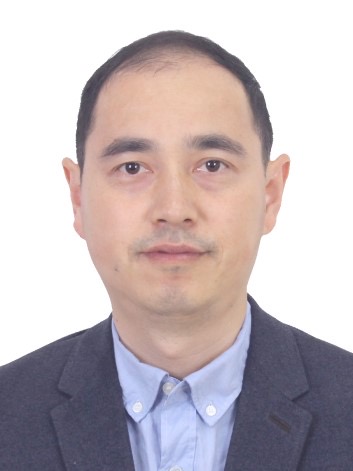}}]{Yongming Huang}
	(M’10-SM’16) received the B.S. and M.S. degrees from Nanjing University, Nanjing, China, in 2000 and 2003, respectively, and the Ph.D. degree in electrical engineering from Southeast University, Nanjing, in 2007.
	
	Since March 2007 he has been a faculty in the School of Information Science and Engineering, Southeast University, China, where he is currently a full professor. He has also been the Director of the Pervasive Communication Research Center, Purple Mountain Laboratories, since 2019. From 2008 to 2009, he was visiting the Signal Processing Lab, Royal Institute of Technology (KTH), Stockholm, Sweden. He has published over 200 peer-reviewed papers, hold over 80 invention patents. His current research interests include intelligent 5G/6G mobile communications and millimeter wave wireless communications. He submitted around 20 technical contributions to IEEE standards, and was awarded a certiﬁcate of appreciation for outstanding contribution to the development of IEEE standard 802.11aj. He served as an Associate Editor for the IEEE Transactions on Signal Processing and a Guest Editor for the IEEE Journal on Selected Areas in Communications. He is currently an Editor-at-Large for the IEEE Open Journal of the Communications Society.
\end{IEEEbiography}


\end{document}